\documentclass{article}

\usepackage[a4paper]{geometry}

\usepackage{graphics,graphicx,color}
\usepackage{enumerate,paralist,hyperref}
\usepackage{caption,subcaption,xspace}
\usepackage{todonotes}
\usepackage{amsmath}
\usepackage{amssymb}
\usepackage{amsfonts}
\usepackage{amsthm}
\graphicspath{{img/}}

\newtheorem{theorem}{Theorem}
\newtheorem{lemma}{Lemma}

\newtheorem{property}{Property}
\newtheorem{definition}{Definition}

\newcommand{\planar}[1]{${#1}$-planar\xspace}
\newcommand{\quasi}[1]{${#1}$-quasiplanar\xspace}
\newcommand{\kplanar}{\planar{k}}
\newcommand{\kquasi}{\quasi{k}}

\newcommand{\kplusquasi}{\quasi{(k+1)}}
\newcommand{\stg}{simple topological graph\xspace}
\newcommand{\kplanarstg}{\planar{k} \stg}
\newcommand{\twoplanarstg}{\planar{2} \stg}

\newcommand{\kplusquasistg}{\quasi{(k+1)} \stg}
\newcommand{\f}{f}
\newcommand{\g}{g}

\newcommand{\inte}{\ensuremath{\operatorname{int}}}
\newcommand{\exte}{\ensuremath{\operatorname{ext}}}
\newcommand{\cs}{\ensuremath{\mathcal{C}_{\mathrm{s}}}}
\newcommand{\cns}{\ensuremath{\mathcal{C}_{\mathrm{ns}}}}
\newcommand{\vns}{\ensuremath{V_{\mathrm{ns}}}}
\newcommand{\E}[1]{\ensuremath{\operatorname{E}(#1)}}
\newcommand{\V}[1]{\ensuremath{\operatorname{V}(#1)}}
\newcommand{\VV}[1]{\ensuremath{V_{#1}}}
\newcommand{\EE}[1]{\ensuremath{#1}}
\newcommand{\R}[1]{\ensuremath{\operatorname{R}(#1)}}
\newcommand{\D}[1]{\ensuremath{\operatorname{D}(#1)}}

\newcommand{\edge}[2]{\ensuremath{#1#2}}

\title{Simple $k$-Planar Graphs are Simple $(k+1)$-Quasiplanar\thanks{Preliminary versions of the results presented in this paper appeared at WG 2017~\cite{DBLP:conf/wg/AngeliniBBLBDLM17} and MFCS 2017~\cite{DBLP:conf/mfcs/0001T17}.}}

\author{
Patrizio~Angelini\footnotemark[2] \and
Michael~A.~ Bekos\footnotemark[2] \and
Franz~J.~ Brandenburg\footnotemark[3] \and
Giordano~{Da~Lozzo}\footnotemark[4] \and
Giuseppe~{Di~Battista}\footnotemark[4] \and
Walter~Didimo\footnotemark[5] \and
Michael~Hoffmann\footnotemark[6]\and
Giuseppe~Liotta\footnotemark[5] \and
Fabrizio~Montecchiani\footnotemark[5] \and
Ignaz~Rutter\footnotemark[3] \hfil 
Csaba~D.~T\'oth\footnotemark[7]~~~\footnotemark[8]
\\[0.2in]
\footnotemark[2]~~~Universit\"at T\"ubingen, T\"ubingen, Germany\\\texttt{{\{angelini,bekos\}@informatik.uni-tuebingen.de}}
\\[0.1in]
\footnotemark[3]~~~University of Passau, Passau, Germany\\\texttt{{\{brandenb,rutter\}@fim.uni-passau.de}}
\\[0.1in]
\footnotemark[4]~~~Roma Tre University, Rome, Italy\\\texttt{{\{giordano.dalozzo,gdb\}@uniroma3.it}}
\\[0.1in]
\footnotemark[5]~~~Universit\'a degli Studi di Perugia, Perugia, Italy\\\texttt{{\{walter.didimo,giuseppe.liotta,fabrizio.montecchiani\}@unipg.it}}
\\[0.1in]
\footnotemark[6]~~~Department of Computer Science, ETH Z\"urich, Z\"urich, Switzerland\\\texttt{{hoffmann@inf.ethz.ch}}
\\[0.1in]
\footnotemark[7]~~~California State University Northridge, Los Angeles, CA, USA\\\texttt{{cdtoth@acm.org}}
\\[0.1in]
\footnotemark[8]~~~Tufts University, Medford, MA, USA
}

\date{} 

\begin{document}

\maketitle

\begin{abstract}
A simple topological graph is $k$-quasiplanar ($k\geq 2$) if it contains no $k$ pairwise crossing edges, and $k$-planar if no edge is crossed more than $k$ times. In this paper, we explore the relationship between $k$-planarity and $k$-quasiplanarity to show that, for $k \geq 2$, every $k$-planar simple topological graph can be transformed into a $(k+1)$-quasiplanar simple topological graph.
\end{abstract}

\clearpage

\section{Introduction}

A \emph{topological graph} is a graph drawn in the plane
such that vertices are mapped to distinct points and
each edge is mapped to a Jordan arc between its endpoints without passing through any other vertex.
In this paper, we only deal with topological graphs containing neither multi-edges nor self-loops.
We do not distinguish between the vertices (resp., edges) and the points (resp., arcs) they are mapped to.
A topological graph is \emph{simple} if any two edges intersect in at most one point,
which is either a common endpoint or a proper crossing.

A topological graph is \emph{$k$-planar}, for $k \geq 0$, if each
edge is crossed at most $k$ times, and \emph{$k$-quasiplanar}, for
$k \geq 2$, if there are no $k$ pairwise crossing edges.
A graph is \emph{$k$-planar} (\emph{$k$-quasiplanar}) if it is isomorphic to
the underlying abstract graph of a $k$-planar ($k$-quasiplanar) topological graph.

A graph is \emph{simple $k$-planar} (\emph{simple $k$-quasiplanar}) if it is isomorphic to
the underlying abstract graph of a simple $k$-planar ($k$-quasiplanar) topological graph.
By definition, a graph is planar if and only if it is $0$-planar ($2$-quasiplanar). Note that
$3$-quasiplanar graphs are also called \emph{quasiplanar}.
Refer to \figurename~\ref{fi:intro1} for examples.

\begin{figure}[t]
	\centering
    \begin{minipage}[b]{.2\textwidth}
        \centering
        \includegraphics[page=1,width=\textwidth]{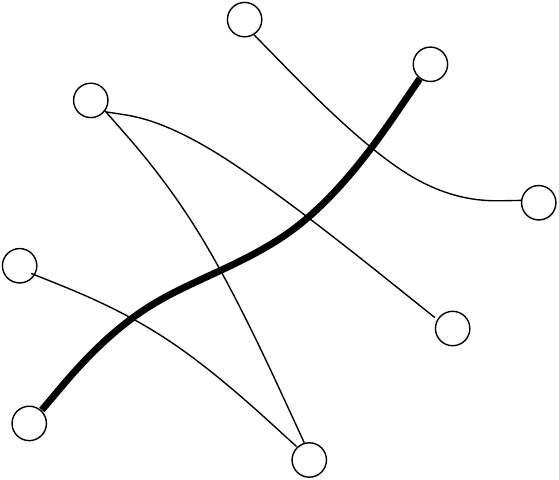}
        \subcaption{~}\label{fi:intro1-a}
    \end{minipage}
	\hfil
	\begin{minipage}[b]{.2\textwidth}
		\centering
		\includegraphics[page=1,width=\textwidth]{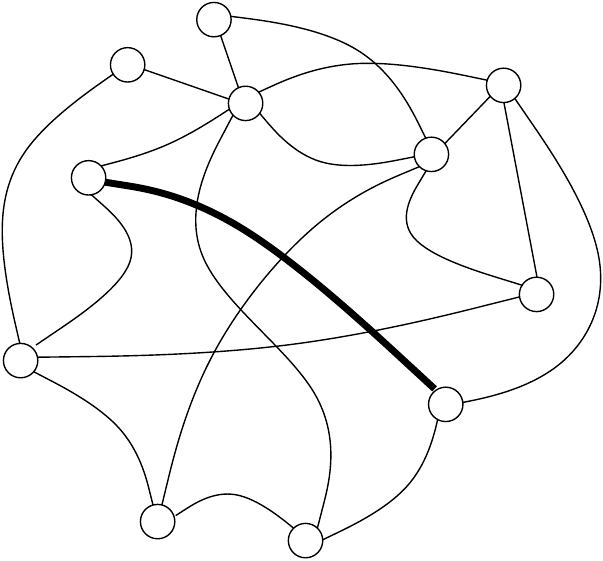}
		\subcaption{~}\label{fi:intro1-b}
	\end{minipage}
	\hfil
	\begin{minipage}[b]{.2\textwidth}
		\centering
		\includegraphics[page=1,width=\textwidth]{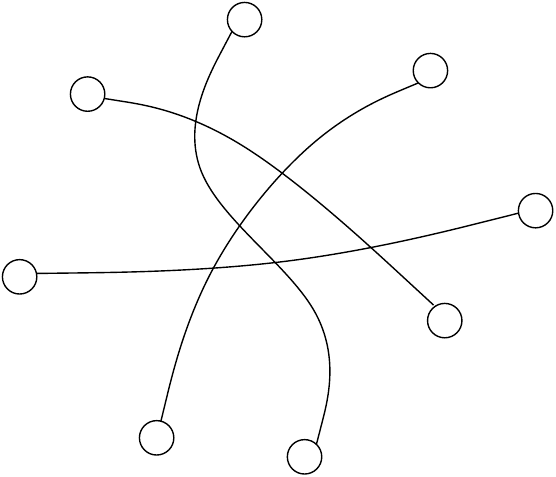}
		\subcaption{~}\label{fi:intro1-c}
	\end{minipage}
	\hfil
	\begin{minipage}[b]{.2\textwidth}
		\centering
		\includegraphics[page=1,width=\textwidth]{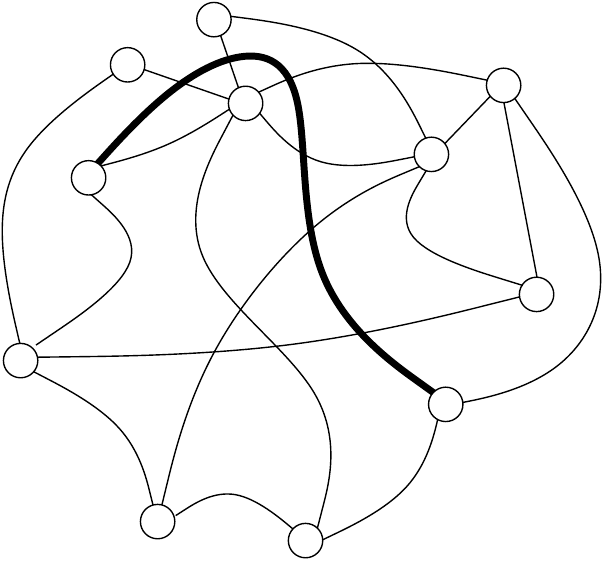}
		\subcaption{~}\label{fi:intro1-d}
	\end{minipage}
	\caption{%
	(a) A crossing configuration that is forbidden in a $3$-planar topological graph.
	(b) A $3$-planar topological graph.
	(c) A crossing configuration that is forbidden in a $4$-quasiplanar topological graph.
	(d) A $4$-quasiplanar topological graph obtained from the one of \figurename~(b) by suitably rerouting the thick edge.}
	\label{fi:intro1}
\end{figure}

The $k$-planar and $k$-quasiplanar graphs are part of a family of
classes of topological graphs that are defined by restrictions on
crossings. Informally, these classes are called \emph{beyond planar}, we refer the interested reader to~\cite{DBLP:journals/csur/DidimoLM19}. Further popular classes are defined by
the exclusion of (natural or radial) grids~\cite{DBLP:journals/comgeo/AckermanFPS14,DBLP:journals/gc/PachPST05}, including fan-planar~\cite{DBLP:journals/corr/KaufmannU14} and fan-crossing free
graphs~\cite{DBLP:journals/algorithmica/CheongHKK15}.

Primarily, the edge density of graphs has been studied for these classes.
By Euler's polyhedron formula, every planar graph on $n\geq 3$ vertices has at most $3n-6$ edges, and this bound is tight.
In fact, for constant $k\geq 0$, every $k$-planar graph is sparse.
Pach and T\'oth~\cite{DBLP:journals/combinatorica/PachT97} proved that a $k$-planar graph with $n$ vertices has at most $4.108\sqrt{k}\,n$ edges\footnote{The upper bound was stated for $k$-planar \emph{simple} topological graphs in~\cite{DBLP:journals/combinatorica/PachT97}, but the proof extends verbatim to all $k$-planar topological graphs.}.
For simple $k$-planar graphs, where $k \leq 4$, Pach and T\'oth~\cite{DBLP:journals/combinatorica/PachT97} also established a finer bound of $(k+3)(n-2)$, and proved that this bound is tight for $k \leq 2$.
For $k=3$ and for $k=4$, the best known upper bounds on the number of edges are $5.5n - 11$ and $6n-12$, respectively, which are tight up to small additive constants~\cite{DBLP:journals/corr/Ackerman15,DBLP:conf/gd/Bekos0R16,DBLP:journals/dcg/PachRTT06}. A consequence of the result in~\cite{DBLP:journals/corr/Ackerman15} is that the upper bound for $k$-planar graphs can be improved to $3.81\sqrt{k}\,n$.

Concerning $k$-quasiplanar graphs, a 20-year-old conjecture by Pach, Shah\-rokhi, and Szegedy asserts that for every $k\geq 2$ there is a constant $c_k$ such that every $k$-quasiplanar graph with $n$ vertices has at most $c_k n$ edges~\cite{DBLP:journals/algorithmica/PachSS96}. However, the conjecture has only been settled for $k=2,3,4$. Agarwal et al.~\cite{DBLP:journals/combinatorica/AgarwalAPPS97} were the first to prove that simple $3$-quasiplanar  graphs have a linear number of edges. This was generalized by Pach et al.~\cite{DBLP:conf/jcdcg/PachRT02}, who proved that \emph{every} $3$-quasiplanar graph on~$n$ vertices has at most $65n$ edges. This bound was further improved to $8n-O(1)$ by Ackerman and Tardos~\cite{DBLP:journals/jct/AckermanT07}. For simple $3$-quasiplanar graphs they also proved a bound of $6.5n-20$, which is tight up to an additive constant. Ackerman~\cite{DBLP:journals/dcg/Ackerman09} proved that $4$-quasiplanar graphs have at most a linear number of edges.
For $k \geq 5$, several authors have shown super-linear upper bounds on the number of edges in $k$-quasiplanar graphs (see, e.g.,~\cite{DBLP:journals/jct/CapoyleasP92,DBLP:conf/compgeom/FoxP08,DBLP:journals/siamdm/FoxPS13,DBLP:journals/algorithmica/PachSS96,DBLP:journals/dcg/Valtr98}). The most recent results are due to Suk and Walczak~\cite{DBLP:journals/comgeo/SukW15}, who proved that every $k$-quasiplanar simple topological graph on $n$ vertices has at most $c_k' n \log n$ edges, where $c_k'$ depends only on $k$.  For $k$-quasiplanar topological graphs where two edges can cross in at most $t$ points, they give an upper bound of $2^{\alpha(n)^c} n \log n$, where $\alpha(n)$ is the inverse of the Ackermann function, and $c$ depends only~on~$k$~and~$t$.

Note that every simple $k$-planar graph is simple $(k+1)$-planar, and every simple $k$-quasiplanar graph is simple $(k+1)$-quasiplanar, by definition.
It is not difficult to see that the converse is false in both cases.
This naturally defines a hierarchy of $k$-planarity and a hierarchy of \mbox{$k$-quasiplanarity}.
However, the relation between the hierarchies is not fully understood.
For every $k\geq 3$, there are infinitely many simple $3$-quasiplanar graphs that are not simple $k$-planar~\cite{BAE2018}.
Also, it is easy to see that, for $k \geq 1$, every $k$-planar simple topological graph is $(k+2)$-quasiplanar. Indeed, if a $k$-planar simple topological graph $G$ were not  $(k+2)$-quasiplanar, it would have $k+2$ pairwise crossing edges,
each of which crosses at least $k+1$ other edges, thus contradicting the hypothesis that $G$ is $k$-planar.

\medskip

\noindent{\bf Contribution.} In this paper we focus on simple topological graphs and prove a notable inclusion relationship between the $k$-planarity and the $k$-quasiplanarity hierarchies. The proof is constructive and sheds a new light on the structure of $k$-planar and $k$-quasiplanar simple topological graphs.
We show that every simple $k$-planar graph is simple $(k+1)$-quasiplanar for every $k \geq 2$.
More precisely, we show that a $k$-planar simple topological graph, with $k \ge 2$, can be transformed into an isomorphic simple topological graph that contains no $k+1$ pairwise crossing edges (although an edge may be crossed more than $k$ times). For example, the simple topological graph in \figurename~\ref{fi:intro1-b} is $3$-planar but not $4$-quasiplanar. By rerouting an edge, we obtain the simple topological graph in \figurename~\ref{fi:intro1-d}, which is $4$-quasiplanar (but not $3$-planar).
Note that this result cannot be extended to the case $k=1$, as a $2$-quasiplanar graph is planar.

The proof of our result is based on the following novel methods:
\emph{(i)} A general-purpose technique to ``untangle'' a set of pairwise crossing edges.
More precisely, we show how to reroute the edges of a $k$-planar simple topological graph in
such a way that all vertices of a set of $k+1$ pairwise crossing edges lie in
the same connected region of the plane (that is, a face of the arrangement induced
by the edges). \emph{(ii)} A global edge rerouting technique, whose main ingredients are a matching argument and a systematic study of the cycles in an auxiliary ``conflict'' graph, used to remove all forbidden configurations of $(k+1)$ pairwise crossing edges from a $k$-planar simple topological graph, provided that these edges are ``untangled.''

\medskip

\noindent{\bf Paper organization.} The remainder of the paper is structured as follows. In Section~\ref{sec:basic} we give some basic terminology, we describe the untangling procedure in \emph{(i)}, and we prove important properties of the resulting topological graphs. Section~\ref{se:rerouting-strategy} outlines our general proof strategy. Section~\ref{sec:function-f} shows how to compute a global rerouting as described in \emph{(ii)} that results in a $(k+1)$-quasiplanar topological graph, for $k \ge 3$. Section~\ref{se:properties-rerouted} proves properties of a rerouted topological graph that are useful to prove both $(k+1)$-quasiplanarity when $k=2$ and simplicity in Sections~\ref{se:kquasiplanarity} and~\ref{se:simplicity}, respectively.
In particular, Section~\ref{se:kquasiplanarity} also contains a more sophisticated argument to compute a suitable global rerouting when $k=2$. Conclusions and open problems are in Section~\ref{sec:conclusions}.

\section{Basic Tools and Properties}\label{sec:basic}

We first state further basic definitions and notation that will be used throughout the paper.
As already stated, we only consider graphs with neither parallel edges nor self-loops. Also, we
assume our graphs to be connected, as our results immediately carry over to disconnected graphs.
In notation and terminology, we do not distinguish between the vertices (edges) of a topological graph
and the points (Jordan arcs) representing them.
Recall that a topological graph is simple if any two edges share at most one point, which is either a common endpoint or a proper crossing.
A topological graph is \emph{almost simple} if any two edges share at most one internal point, which is a proper crossing
(i.e., pairs of adjacent edges may cross but at most once).
For a topological graph~$G$, the set $\mathbb{R}^2\setminus G$ is open, and its connected components are called \emph{faces}.
The unique unbounded face is the \emph{outer face}, any bounded face is an \emph{inner face}.
 Note that the boundary of a face can contain vertices of the graph and crossing points between edges.

For two graphs or topological graphs, $G$ and $G'$, we write $G \simeq G'$ if they are isomorphic
or their underlying abstract graphs are isomorphic.
A graph $G$ is \kplanar (\kquasi) if there exists a \kplanar (\kquasi) topological graph $G'$ such that $G\simeq G'$.

\begin{figure}[t]
	\centering
    \begin{minipage}[b]{.25\textwidth}
		\centering
		\includegraphics[page=2,width=\textwidth]{bundle-2.pdf}
		\subcaption{~}\label{fi:bad-bundle}
	\end{minipage}
	\hfil
	\begin{minipage}[b]{.25\textwidth}
        \centering
        \includegraphics[page=1,width=\textwidth]{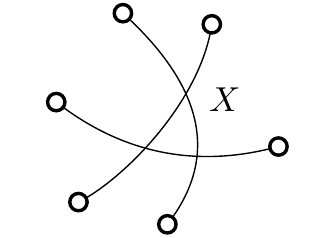}
        \subcaption{~}\label{fi:good-bundle}
    \end{minipage}
    \hfil
	\begin{minipage}[b]{.25\textwidth}
        \centering
        \includegraphics[page=3,width=\textwidth]{bundle-2.pdf}
        \subcaption{~}\label{fi:k-gon}
    \end{minipage}
	\caption{%
	(a)~A tangled $3$-crossing; the circled vertices and the solid vertices belong to different faces of the arrangement.
	(b)~An untangled $3$-crossing; all vertices belong to the same face of the arrangement (the outer face).
	(c)~The $6$-gon spanned by the $3$-crossing in (b). }
	\label{fi:bundles}
\end{figure}
Let $G=(V,E)$ be a simple topological graph and let $k\geq 2$ be an integer.
A \emph{fan} of $G$ is a set of edges that share a common endpoint.
A set $X\subset E$ of $k$ pairwise crossing edges is called a \emph{$k$-crossing}.
Note that the edges in $X$ are pairwise non adjacent since $G$ is a simple topological graph.
For a $k$-crossing $X$, denote by $\V{X}$ the set of $2k$ endpoints of the $k$ edges in $X$.
The \emph{arrangement of $X$}, denoted by $A_X$, is the arrangement of the Jordan arcs in $\EE{X}$.
A \emph{node} of $A_X$ is either a vertex or a crossing point of two edges in $\EE{X}$.
A \emph{segment} of $A_X$ is a part of an arc in $\EE{X}$ between two consecutive nodes
(i.e., a maximal uncrossed part of an edge in $\EE{X}$).
A $k$-crossing $X$ is \emph{untangled} if in the arrangement $A_X$ all $2k$ vertices in $\V{X}$ are incident to a common face. Otherwise, it is \emph{tangled}. For example, the $3$-crossing in \figurename~\ref{fi:bad-bundle} is tangled, whereas the $3$-crossing in \figurename~\ref{fi:good-bundle} is untangled. We observe the following.

\begin{property}\label{pr:k-crossings-disjoint}
Let $G=(V,E)$ be a \kplanarstg topological graph and let $X$ be a $(k+1)$-crossing in $G$.
An edge in $\EE{X}$ cannot be crossed by any other edge in $E \setminus \EE{X}$.
Consequently, for any two distinct $(k+1)$-crossings $X$ and $Y$ in $G$, we have $\EE{X} \cap \EE{Y} = \emptyset$.
\end{property}
\begin{proof}
Each edge $e$ in a $(k+1)$-crossing $X$ crosses each of the remaining $k$ edges in $\EE{X}$.
Since graph $G$ is $k$-planar, edge $e$ is not crossed by any other edge in $ E \setminus \EE{X}$.
\end{proof}

In the next subsection we show that tangled $(k+1)$-crossings can always be removed.

\subsection{Eliminating tangled $(k+1)$-crossings}\label{ssec:untangling}

The proof of the next lemma describes how to ``untangle'' all $(k+1)$-crossings in a \kplanarstg.
This method is of general interest, as it gives more insights on the structure of {\kplanarstg}s.

\begin{lemma}\label{lem:removing-tangled}
Let $G$ be a \kplanarstg. There exists a \kplanarstg $G'$, $G' \simeq G$, without tangled $(k+1)$-crossings.
\end{lemma}
\begin{proof}
We first show how to untangle a $(k+1)$-crossing $X$ in a \kplanarstg~$G$ while neither creating new $(k+1)$-crossings nor introducing new crossings.

Let $X$ be a tangled $(k+1)$-crossing and let $A_X$ be its arrangement. For each face $f$ of $A_X$, denote by $\VV{f}$ the set of nodes in $\V{X}$ incident to $f$. Since every vertex in $\VV{X}$ is incident precisely to one arc in $X$, the set $\V{X}$ is partitioned into subsets $\VV{f}$ over all faces $f$ of $A_X$.

\begin{figure}[t]
	\centering
    \begin{minipage}[b]{.4\textwidth}
        \centering
        \includegraphics[page=1,width=0.9\textwidth]{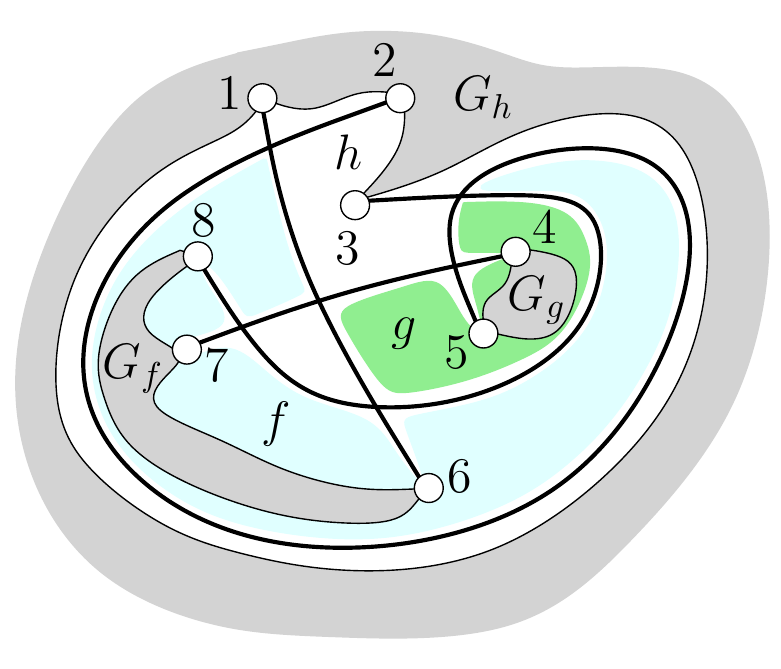}
        \subcaption{~}\label{fi:untangling-a}
    \end{minipage}
	\hfil
	\begin{minipage}[b]{.35\textwidth}
		\centering
		\includegraphics[page=1,width=0.9\textwidth]{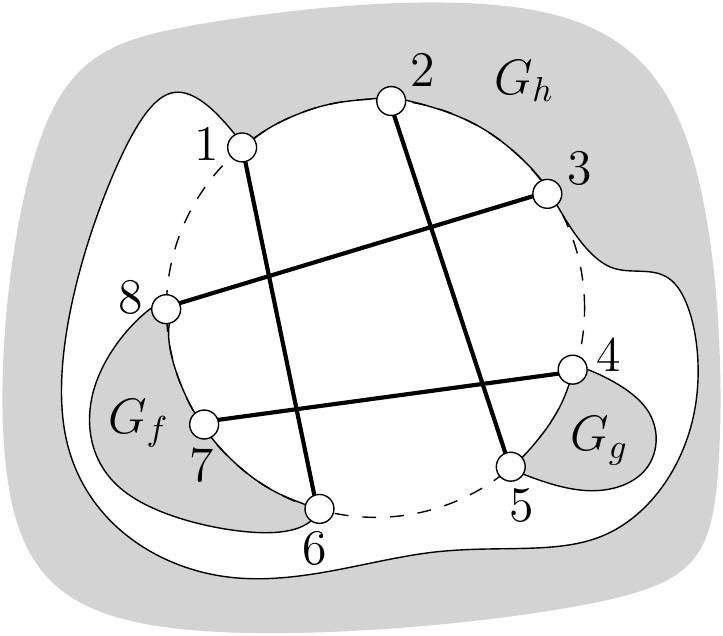}
		\subcaption{~}\label{fi:untangling-b}
	\end{minipage}
	\caption{%
	Illustration of the untangling procedure in the proof of Lemma~\ref{lem:removing-tangled}:
	(a)~A $3$-planar simple topological graph with a $4$-crossing $X$ (thicker edges).
	(b)~The topological graph resulting from the procedure that untangles~$X$.}
	\label{fi:untangling}
\end{figure}

For every face $f$ of $A_X$, denote by $G_f$ the subgraph of $G$ consisting of the vertices of $\VV{f}$, and of all vertices and edges of $G$ that lie in the interior of $f$. Refer to \figurename~\ref{fi:untangling-a} for an illustration.
%
%
By Property~\ref{pr:k-crossings-disjoint}, every edge in $E\setminus \EE{X}$ lies in a face of $A_X$.
Consequently, the topological graph $(V,E\setminus \EE{X})$ is the disjoint union of the graphs $G_f$,
each of which is $k$-planar.

For every inner face $f$, there is a region $D_f\subset f$ homeomorphic to an open disk such that its boundary contains $\VV{f}$, and all other vertices and edges in $G_f$ lie in $\inte{D_f}$. For the outer face $h$, there is a region $D_h$ homeomorphic to the complement of a closed disk such that its boundary contains $V_h$, and all other vertices and edges lie in $\inte{D_h}$.

We construct a topological graph $G'$, $G'\simeq G$, as follows (see \figurename~\ref{fi:untangling-b}). Let $C$ be a circle in the plane.
For every face $f$ of $A_X$ (including the outer face), apply a homeomorphism that maps the region $D_f$ to some region in the exterior of $C$ 
such that a Jordan arc in $\partial D_f$ that contains $V_f$ is mapped into $C$, and resulting regions are pairwise disjoin. Draw the $k+1$ edges in $X$ as straight-line segments in the interior of $C$. Each subgraph $G_f$ of $G$ is mapped to a $k$-planar topological graph $G_f'$, $G_f'\simeq G_f$. The $(k+1)$-crossing $X$ is mapped to a set $X'$ of $k+1$ edges that is either not a $(k+1)$-crossing or an untangled $(k+1)$-crossing. Since two edges in $G'$ cross only if the corresponding edges cross in $G$, the topological graph $G'$ is simple and $k$-planar, and no new $(k+1)$-crossing is created.

Successively apply the above transformation as long as it contains a tangled $(k+1)$-crossing. Since the number of tangled $(k+1)$-crossings decreases, we eventually obtain a $k$-planar topological graph $G'$, $G'\simeq G$, without tangled $(k+1)$-crossings.
\end{proof}

\subsection{Properties of untangled $(k+1)$-crossings}\label{ssec:properties-untangled}

Let $G_0$ be a simple $k$-planar graph. We wish to show that $G_0$ is simple $(k+1)$-quasiplanar.
We may assume that $G_0$ is \emph{edge-maximal} (i.e., the addition of any edge would yield a graph that is not simple $k$-planar).
Let $G$ be a simple $k$-planar topological graph such that $G\simeq G_0$. We may assume that $G$ is \emph{crossing minimal} (i.e., $G$ has the minimum number of edge crossings over all \kplanarstg{s} isomorphic to $G_0$), and that every $(k+1)$-crossing is untangled by Lemma~\ref{lem:removing-tangled}. We may assume, by applying a projective transformation if necessary, that for every $(k+1)$-crossing $X$,
all vertices of $V(X)$ are incident to the outer face of $A_X$.

Then every (untangled) $(k+1)$-crossing $X$ in $G$ spans a (topological) \emph{$2(k+1)$-gon} in the following sense. All $2(k+1)$ vertices of $\V{X}$ lie on a face $f_X$ of the arrangement $A_X$ induced by the edges of $X$ as drawn in $G$. Any two vertices of $\V{X}$ that are consecutive along the boundary of $f_X$ can be connected by a Jordan arc that closely follows the boundary of $f_X$ and does not cross any edge in $G$; see \figurename~\ref{fi:k-gon}. Together these arcs form a closed Jordan curve, which partitions the plane into two connected regions: let $\R{X}$ denote the closed region homeomorphic to a disk that contains the edges of $X$, and let $\partial\R{X}$ denote the boundary of $\R{X}$. We think of $\partial\R{X}$ as both a closed Jordan curve and as a topological graph that is a $2(k+1)$-cycle. Let $\mathcal X$ be the set of all $(k+1)$-crossings of $G$. By Property~\ref{pr:k-crossings-disjoint}, we may assume that for every $X,X'\in \mathcal{X}$, $X\neq X'$, the regions $\R{X}$ and $\R{X'}$ do not share any interior point. The following observation holds.

\begin{property}\label{pr:gonedge}
For each $X \in \mathcal X$, every pair of consecutive vertices of the $2(k+1)$-cycle $\partial\R{X}$ are connected by an edge in $G$, which is crossing-free.
\end{property}
\begin{proof}
Let $X \in \mathcal X$, and let $u,v\in V$ be two consecutive vertices of the $2(k+1)$-cycle $\partial\R{X}$.  We show that $\edge{u}{v}$ is an edge in $G$. Indeed, if $\edge{u}{v}$ is not an edge of $G$, we can augment $G$ by drawing this edge  as a crossing-free Jordan arc along $\partial\R{X}$ (without violating $k$-planarity). This contradicts our assumption that $G$ is edge-maximal, and thus proves that $\edge{u}{v}$ is an~edge~in~$G$.

We then show that $\edge{u}{v}$ is crossing free in $G$. Indeed, if it crossed any other edge in $G$, we could redraw it as a crossing-free Jordan arc along $\partial\R{X}$, obtaining a topological graph that is still $k$-planar but with fewer crossings than $G$. This  contradicts our assumption that $G$ is crossing minimal.
\end{proof}

By Property~\ref{pr:gonedge} any two consecutive vertices along the boundary $\partial\R{X}$ of a $2(k+1)$-gon $\R{X}$ are connected
by an edge $e$ in $G$. Note that this does not necessarily imply that $e$ is drawn along $\partial\R{X}$. It is possible that the cycle formed by the edge $e$ in $G$ and the portion of $\partial\R{X}$ connecting the endpoints of $e$ contains other parts of the graph.

\begin{property}\label{pr:kgon}
\begin{enumerate}[(a)]
\item[]
\item\label{obs:kgon:1} Let $X_1,X_2\in \mathcal X$ such that $X_1 \neq
  X_2$. Then $\V{X_1}$ and $\V{X_2}$ share at most $2k+1$
  vertices.
\item\label{obs:kgon:2} Let $X_1$, $X_2$, and $X_3$ be three pairwise distinct
  $(k+1)$-crossings in $\mathcal X$. Then $\V{X_1}$, $\V{X_2}$,
  and $\V{X_3}$ share at most two vertices.
\end{enumerate}
\end{property}
\begin{proof}
  \eqref{obs:kgon:1} Suppose that $\partial\R{X_1}$ and $\partial\R{X_2}$ share $2k+2$ vertices. Since $\R{X_1}$ and $\R{X_2}$ are contractible and interior-disjoint, the counterclockwise order of the vertices along $\partial\R{X_1}$ and  $\partial\R{X_2}$, respectively, are reverse to each other. Every edge in $X_i$, for $i\in\{1,2\}$, connects antipodal points along  $\partial\R{X_i}$. Antipodal pairs are invariant under reversal, consequently   every edge in $X_1$ is present in $X_2$, contradicting our assumption that $G$ is a simple graph.

  \eqref{obs:kgon:2} Suppose that $\partial\R{X_1}$, $\partial\R{X_2}$, and $\partial\R{X_3}$ share three distinct vertices $v_1,v_2,v_3$. We obtain a plane drawing of $K_{3,3}$ as follows: Place points $p_1,p_2,p_3$ inside $\R{X_1}$, $\R{X_2}$, $\R{X_3}$, respectively, and connect each of $p_1,p_2,p_3$ to all of  $v_1,v_2,v_3$. All edges incident to $p_i$, for $i\in\{1,2,3\}$, are drawn as  a plane star inside $\R{X_i}$. As the regions $\R{X_1}$, $\R{X_2}$, and $\R{X_3}$ are interior-disjoint, no two edges cross. As $K_{3,3}$ is nonplanar, we obtain a contradiction.
\end{proof}

\section{Edge Rerouting Operations and Proof Strategy}\label{se:rerouting-strategy}

We introduce an edge rerouting operation that will be crucial for our proof strategy.
Let $G$ be a \kplanarstg in which all $(k+1)$-crossings are untangled (Lemma~\ref{lem:removing-tangled}).
Let $X$ be a $(k+1)$-crossing in $G$.
Without loss of generality, the vertices in $\V{X}$ lie in the outer face of $A_X$.

Let $e = \edge{u}{v} \in \EE{X}$ and let $w \in \V{X} \setminus \{u,v\}$ such that $u$ and $w$ are consecutive along $\partial \R{
X}$. Let $\D{X}\subset \R{X}$ be a region homeomorphic to a disk that encloses all crossing points of $X$ and such that each edge in $\EE{X}$ crosses the boundary $\partial \D{X}$ of $\D{X}$ exactly twice. Let $w'$ be the other endpoint of the unique edge in $\EE{X}$ that is incident to $w$. Note that edge $\edge{u}{w}$ is in $E$ by Property~\ref{pr:gonedge}.

The operation \emph{rerouting $e=\edge{u}{v}$ around $w$} consists of redrawing $e$ as follows.
Refer to \figurename~\ref{fig:k-crossing-reroute-nonadjacent} for an illustration.
Starting from vertex $v$, follow the edge $e$ until reaching the first crossing with $\D{X}$; we call this part a \emph{tip} of $e$. Then, follow the shortest path along $\partial \D{X}$ to the crossing of $\edge{w}{w'}$ with $\partial \D{X}$ closer to $w$ (without crossing $\edge{w}{w'}$). Then, follow edge $\edge{w}{w'}$ until vertex $w$, and go around $w$ until reaching edge $\edge{u}{w}$, (in an orientation that avoids crossing $\edge{w}{w'}$); we call this part the \emph{hook} of $e$. Finally, complete the new drawing of $e$ by following $\edge{u}{w}$ to $u$; this part is called a \emph{tip} of $e$ (hence edge $e$ has two tips).

\begin{figure}[b]
	\centering
    \begin{minipage}[b]{.3\textwidth}
        \centering
        \includegraphics[page=4]{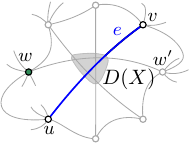}
        \subcaption{~}\label{fig:k-crossing-curve}
    \end{minipage}
	\hfil
	\begin{minipage}[b]{.3\textwidth}
		\centering
		\includegraphics[page=3]{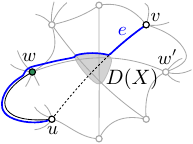}
		\subcaption{~}\label{fig:k-crossing-reroute-nonadjacent}
	\end{minipage}
	\hfil
	\begin{minipage}[b]{.3\textwidth}
		\centering
		\includegraphics[page=5]{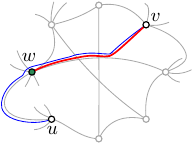}
		\subcaption{~}\label{fig:k-crossing-reroute-home}
	\end{minipage}
	\caption{%
	The rerouting operation for dissolving untangled $(k+1)$-crossings.
	(a)~An untangled $(k+1)$-crossing $X$.
	(b)~The rerouting of edge $\edge{u}{v}$ around the marked vertex $w$.
	(c)~The additional rerouting of $\edge{v}{w}$.}
	\label{fig:k-crossing}
\end{figure}

\begin{lemma}\label{lem:reroute-properties}
Let $G$ be a \kplanarstg and let $X$ be an untangled $(k+1)$-crossing in $G$.
Let $G'$, $G' \simeq G$, be the topological graph obtained from $G$ by rerouting an edge $e=\edge{u}{v} \in \EE{X}$ around a vertex $w \in \V{X} \setminus \{u,v\}$ such that $u$ and $w$ are consecutive along $\partial\R{X}$.
$G'$ has the following properties:
\begin{enumerate}[(i)]\itemsep -2pt
\item edges $e$ and $\edge{w}{w'}$ do not cross;
\item the edges that are crossed by $e$ in $G'$ but not in $G$ form a fan at $w$;
\item\label{reroute-properties-iii} edge $e$ does not cross any edge more than once.
\end{enumerate}
\end{lemma}
\begin{proof}
Properties~(i) and~(ii) immediately follow from the definition of the rerouting operation.
We prove property~(iii). First note that the tip of $e$ incident to $v$ does not cross any edge in $G'$, since $X$ is a $(k+1)$-crossing and all the crossings of edge $e$ before the rerouting lie in the interior of $\D{X}$. The part of $e$ that follows $\partial \D{X}$ crosses all other edges in $X$, except for $\edge{w}{w'}$, exactly once, since the two pairs of crossing points of any two edges of $X$ with $\partial \D{X}$ alternate around $\partial \D{X}$. The hook of $e$ crosses edges that form a fan at $w$, and thus do not belong to $X$, since the only edge of $X$ incident to $w$ is $\edge{w}{w'}$, which is not crossed by $e$. Since these crossings are located in a neighborhood of $w$, none of the edges incident to $w$ is crossed twice by $e$. Finally, the tip of $e$ incident to $u$ follows edge $\edge{u}{w}$, which is crossing-free by Property~\ref{pr:gonedge}. This concludes the proof.
\end{proof}

\paragraph{Homes and home rerouting} In addition, we observe that edge $\edge{v}{w}$,
if it exists, can always be (re)drawn inside $\R{X}$ so that it crosses
neither $\edge{u}{v}$ nor $\edge{w}{w'}$ and thus with at most $k-1$ crossings, by
following the first three parts of the edge $\edge{u}{v}$, i.e., the tip of $\edge{u}{v}$
incident to~$v$, the part of $\edge{u}{v}$ that follows $\partial \D{X}$, and part of the hook
of $\edge{u}{v}$ until $w$; see \figurename~\ref{fig:k-crossing-reroute-home}.
Specifically, if we have rerouted edge $\edge{u}{v}$ around $w$, and edge $\edge{v}{w}$ exists,
we say that $X$ is a \emph{home} for the edge $\edge{v}{w}$; and we call
\emph{home rerouting} the redrawing operation described above.

\paragraph{Global rerouting and full rerouting} In the following we describe our general strategy
for transforming a \kplanarstg $G$ into a \stg $G'$, $G' \simeq G$, that is
\kplusquasi. The idea is to appropriately define two injective functions
$\f : \mathcal{X} \rightarrow V$ and $\g : \mathcal{X} \rightarrow E$, which
associate every $(k+1)$-crossing $X$ in $G$ with a vertex $\f(X) \in \V{X}$ and
with an edge $\g(X) \in X$, respectively, such that an endpoint of $\g(X)$ and
$\f(X)$ are consecutive along~$\partial \R{X}$. Then, we apply
the rerouting operation for all pairs $(\g(X),\f(X))$, i.e., rerouting $\g(X)$
around $\f(X)$. This operation, which we call \emph{global rerouting}, and
denote by $(\g,\f)$, is well-defined since the $(k+1)$-crossings are pairwise
edge-disjoint by Property~\ref{pr:k-crossings-disjoint}.

After this operation, for each edge $e \in E$ that has a home, we perform
a home rerouting operation. Note that every $(k+1)$-crossing in $\mathcal{X}$
is a home for at most one edge, but an edge can have up to two homes (one for each endpoint,
no more because $\f$ is injective). If an edge has two homes, we pick one of them arbitrarily
for the home rerouting operation.

The combined operation consisting of a global rerouting $(\g(X),\f(X))$ and all possible home reroutings will be called \emph{full rerouting} in the following.

\paragraph{Challenges} There are, however, two potential problems
that have to be addressed in order for a global rerouting to eliminate $(k+1)$-crossings.
First, Lemma~\ref{lem:reroute-properties} does not guarantee that the topological graph
obtained by rerouting a single edge $e=(u,v)$ around a vertex $w$ is simple.
Indeed, assuming that $u$ and $w$ are consecutive along $\partial \R{X}$,
if the edge $\edge{v}{w}$ were present in $G$, then the rerouted
edge $e=\edge{u}{v}$ would cross such an edge. This crossing between adjacent edges
may be solved by a home rerouting operation for $\edge{v}{w}$ inside $\R{X}$. However, 
if $\edge{v}{w}$ has a home other than $X$ and we picked this other home to reroute
$\edge{v}{w}$ inside, this crossing would not be avoided. Furthermore, rerouting many
edges simultaneously may create new $(k+1)$-crossings. Both problems
can be solved by suitably choosing functions $\f$ and $\g$.

\paragraph{Outlook} In the next section we start by proving that function $\f$
can be chosen to be injective. Note that, given a function $\f$, choosing $\g$
to be injective is trivial (by Property~\ref{pr:k-crossings-disjoint}). We will
prove in Section~\ref{se:kquasiplanarity} that the injectivity of $\f$ and $\g$
is sufficient to guarantee that the resulting topological graph does not contain
$k+1$ mutually crossing edges, for $k \ge 3$. The case $k=2$ is more challenging
as new $3$-crossings may appear after rerouting the edges of $\g$ around the
vertices of $\f$. To avoid these situations, we modify function $\f$ and more
carefully define function~$\g$, as discussed in
Section~\ref{se:kquasiplanarity}. Finally, we show how to avoid crossings
between adjacent edges so as to obtain a simple topological drawing in
Section~\ref{se:simplicity}.

\section{Computing an injective function $\f$}\label{sec:function-f}

In this section we show the existence of a global rerouting such that no two edges of a $k$-planar topological graph $G$, $k\geq 2$, are rerouted around the same vertex (Lemma~\ref{lem:global}), that is, $\f$ is injective. Note that this condition is also necessary for simplicity; see \figurename~\ref{fig:nonsimple}. We start by defining a bipartite graph composed of the vertices of $G$ and of its $(k+1)$-crossings, and by showing that a matching covering all the $(k+1)$-crossings exists. A bipartite graph with vertex sets $A$ and $B$ is denoted by $H=(A \cup B,\widehat{E})$, where $\widehat{E} \subseteq A \times B)$.  A \emph{matching from} $A$ {\em into} $B$ is a set $M \subseteq E$ such that each vertex in $A$ is incident to exactly one edge in $M$ and each vertex in $B$ is incident to at most one edge in $M$. For a subset $A' \subseteq A$, we denote by $N(A')$ the set of all vertices in $B$ that are adjacent to a vertex in $A'$. We recall that, by Hall's theorem, graph $H$ has a matching from $A$ into $B$ if and only if $|N(A')| \ge |A'|$ for every set $A' \subseteq A$.

\begin{figure}[b]
	\centering
\begin{minipage}[b]{.24\textwidth}
        \centering
        \includegraphics[page=1,width=\textwidth]{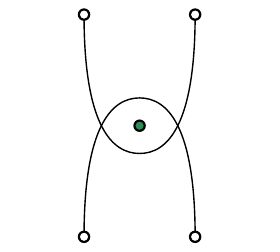}
        \subcaption{~}\label{fig:nonsimple}
    \end{minipage}
    \hfil
    \begin{minipage}[b]{.24\textwidth}
        \centering
        \includegraphics[page=3,width=\textwidth]{simple-newcrossing.pdf}
        \subcaption{~}\label{fig:newcrossing-2}
    \end{minipage}
	\hfil
	\begin{minipage}[b]{.24\textwidth}
		\centering
		\includegraphics[page=2,width=\textwidth]{simple-newcrossing.pdf}
		\subcaption{~}\label{fig:newcrossing-3}
	\end{minipage}
	
	\caption{%
(a) Two edges rerouted around the same vertex. (b)--(c) Two possible cases in which two edges do not cross before a global rerouting operation but cross afterwards. The vertices used for rerouting are filled green.
}
	\label{fig:newcrossing}
\end{figure}

Let $G$ be a \kplanarstg and let $\mathcal{X}$ be the set of $(k+1)$-crossings of $G$. We define a bipartite graph $H = (A \cup B,\widehat{E})$ as follows. For each $(k+1)$-crossing $X \in \mathcal{X}$, set $A$ contains a vertex $v(X)$ and set $B$ contains the endpoints of $\EE{X}$ (that is, $B = \bigcup_{X \in \mathcal{X}} \V{X}$). Also,~$\widehat{E}$ contains an edge between a vertex $v(X) \in A$ and a vertex $u \in B$ if and only if $u \in \V{X}$. We have the following.

\begin{lemma}\label{lem:graphh}
Graph $H = (A \cup B, \widehat{E})$ is a simple bipartite planar graph. Also, each vertex in $A$ has degree $2k+2$.
\end{lemma}
\begin{proof}
    The graph is simple and bipartite by construction. Also, for each $(k+1)$-crossing $X$, vertex $v(X) \in A$ is incident to the $2k+2$ vertices in $B$ belonging to~$\V{X}$.
	We prove that $H$ is also planar by showing that a planar embedding of $H$ can be obtained from $G$ as follows. First, we remove from $G$ all the vertices and edges that are not in any $(k+1)$-crossing. Then, for each $(k+1)$-crossing $X$ of $G$, we remove the portion of $G$ in the interior of $\D{X}$ that encloses all crossings among the edges in $X$ and such that each edge in $\EE{X}$ crosses the boundary of $\D{X}$ exactly twice (as defined in Section~\ref{se:rerouting-strategy}) and add vertex $v(X)$ inside $\D{X}$. Finally, for each vertex $v\in \V{X}$, let $e_v$ be the edge in $X$ incident to $v$ and let $p_v$ be the intersection point between $\partial \D{X}$ and $e_v$ closer to $v$. We complete the drawing of edge $\edge{v(X)}{v}$ by adding an arc between~$v(X)$ and $p_{v}$ in the interior of $\D{X}$ without introducing any crossing. The resulting topological graph is crossing-free.
\end{proof}

\begin{lemma}\label{lem:matching2}
For every nonempty subset $A' \subseteq A$, we have $|N(A')| \ge |A'| + 2k+1$.
\end{lemma}
\begin{proof}
If $|A'|=1$, then $|A'|+2k+1=2k+2$ and a single $(k+1)$-crossing has $2k+2$ vertices. If $|A'| = 2$, then $|A'|+2k+1 =2k+3$; and two distinct $(k+1)$-crossings jointly have at least $2k+3$ vertices by Property~\ref{pr:kgon}\eqref{obs:kgon:1}.
Hence, in both cases the statement holds. Consider now the case $|A'| \ge 3$. Let $H'$ be the subgraph of $H$ induced by $A'\cup N(A')$. Since every vertex in $A$ has degree $2k+2$, by Lemma~\ref{lem:graphh} we have $|\E{H'}| = (2k+2) |A'|$. Also, since $H$ (and thus $H'$) is bipartite planar, by Lemma~\ref{lem:graphh} we have $|\E{H'}| \le 2 (|A'| + N(A')) - 4$. Thus, $|N(A')| \ge k |A'| + 2 = |A'| + (k-1)|A'|+2 \ge |A'| + 3k - 3 + 2 \ge |A'| + 2k+1$, and the statement follows.
\end{proof}

We can now exploit Lemma~\ref{lem:matching2} and Hall's theorem in order to define $\f$ as an injective function, which implies that any corresponding global rerouting is such that no two edges are rerouted around the same vertex.

\begin{lemma}\label{lem:global}
Let $G=(V,E)$ be a \kplanarstg, and let $\mathcal{X}$ be the set of $(k+1)$-crossings of $G$.
It is possible to define a global rerouting $(\g,\f)$ such that $\f$ is injective.
\end{lemma}
\begin{proof}
By Lemma~\ref{lem:matching2} and Hall's theorem, the graph $H$ defined above admits a matching from $A$ into $B$.
For every $(k+1)$-crossing $X\in \mathcal X$, let $\f(X) \in \V{X}$ be the vertex matched to $v(X)$.
It follows that $f:\mathcal{X}\rightarrow V$ is injective. The statement follows by choosing $\g(X)$, for each $(k+1)$-crossing $X$, as one of the two edges in $X$ not incident to $\f(X)$ and with an endpoint adjacent to~$\f(X)$ along $\partial \R{X}$.
\end{proof}

\section{Properties of a rerouted topological graph}\label{se:properties-rerouted}

Let $G'$ be the topological graph obtained from a \kplanarstg $G$ after applying
a full rerouting operation (in which the functions $\f$ and $\g$ are injective).
We study properties of $G'$. In particular, the edges of $G'$ fall into three categories,
depending on how they are represented in $G'$ with respect to $G$:
(1) \emph{nonrerouted edges} have not been rerouted and remain the same as in $G$;
(2) edges that have been rerouted in home rerouting operation we call \emph{safe}
(regardless of whether or not they have also been rerouted in the global rerouting
operation); and
(3) edges that have been rerouted in the global rerouting operation but not in a
home rerouting are \emph{critical}. An edge is \emph{rerouted} if it is either safe
or critical. Let us start by classifying the new crossings that are introduced
by the rerouting algorithm.

\begin{lemma}\label{lem:cross}
Consider two edges $e_1$ and $e_2$ that cross each other in $G'$ but not in $G$. After possibly exchanging the roles of $e_1$ and $e_2$, 
one of the following holds:
\begin{itemize}\itemsep -2pt
\item[(a)] $e_1$ is safe and $e_2$ is a nonrerouted edge of the home $X$ of $e_1$ in which $e_1$ has been rerouted; 
\item[(b)] $e_1$ is critical and rerouted around an endpoint of $e_2$.
\end{itemize}
\end{lemma}
\begin{proof}
  Suppose first that at least one of $e_1$ and $e_2$ is critical, say $e_1$.
  Since $e_1$ is critical, there is a $(k+1)$-crossing $X \in \mathcal X$ such that
  $e_1=\g(X)=(u,v)$ is rerouted around the vertex $w=\f(X)$ in the global rerouting operation.
  We can assume that $u$ is adjacent to $w$ along $\partial\R{X}$.
  Recall that the tip of $e_1$ incident to $v$ is the same in $G'$ as in $G$,
  while the tip of $e_1$ incident to $u$ follows an edge that is crossing-free in $G$.
  Also, $e_1$ does not cross any edge of $G'$ that has been rerouted in $X$ by a home rerouting operation.
  Then, if $e_2$ is a nonrerouted edge or a safe edge, $e_2$ is incident to $w$ and $e_1$ crosses $e_2$ with its hook,
  thus case (b) of the statement applies.
  If $e_2$ is critical, then the hook of $e_2$ does not cross the hook of $e_1$ because $\f$ is injective,
  and thus we are in case (b) of the statement again, as either the hook of $e_1$ crosses a tip of $e_2$ or vice-versa.

  Suppose now that none of $e_1$ and $e_2$ is critical.
  Since $e_1$ and $e_2$ do not cross in $G$, at least one of them is safe, say $e_1$.
  Let $X \in \mathcal{X}$ be the home of $e_1$ in which $e_1$ has been rerouted.
  If $e_2$ is nonrerouted, then case (a) of the statement applies.
  On the other hand, $e_2$ cannot be safe, as otherwise it would be rerouted inside another $2(k+1)$-gon $\R{X'}$,
  which would be interior-disjoint from $\R{X}$, and thus $e_1$ and $e_2$ would not cross each other.
\end{proof}

\begin{lemma}\label{lem:nosafeedges}
  Let $e$ be a safe edge in $G'$ and let $X \in \mathcal X$ be the home in which $e$ has been rerouted.
  Assume that $e=(\f(X),z)$. The following properties hold:
  \begin{enumerate}[(i)]\itemsep -2pt
  \item\label{nosafeedges-i} $e$ is not part of a $(k+1)$-crossing;
  \item\label{nosafeedges-ii} $e$ does not cross any edge more than once; and
  \item\label{nosafeedges-iii} $e$ crosses an adjacent edge $e'$ only if $e'$ is critical, incident to
    $\f(X)$, and rerouted around $z$, where $z=\f(X')$ for some $(k+1)$-crossing
    $X'\in \mathcal X\setminus\{X\}$ with $\g(X')=e'$;
    see \figurename~\ref{fig:nosafeedges} for an illustration.
  \end{enumerate}
\end{lemma}

\begin{figure}[htbp]
	\centering
		\includegraphics[page=2]{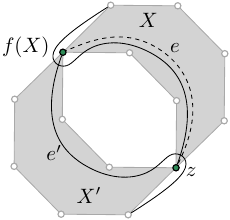}
                \caption{A safe edge $e$ crosses an adjacent edge $e'$; see
                  Lemma~\ref{lem:nosafeedges}(iii).\label{fig:nosafeedges}}
\end{figure}

\begin{proof}
  By definition, $e$ is the only edge that has $X$ as a home. By
  Lemma~\ref{lem:cross} there are only two types of crossings involving $e$:

  (a) inside $\R{X}$ the edge $e$ crosses only the edges of $X$ that have not
  been rerouted and it does not cross the edge incident to $\f(X)$. Since
  $\g(X)$ has been rerouted, $e$ crosses at most $k-1$ such edges. Also, $e$
  crosses these 
  edges only once and it is not adjacent to any of them.

  (b) $e$ crosses an edge $e'$ that has been rerouted around an endpoint of
  $e$. However, by construction $e$ does not cross the edge $\g(X)$ that is
  rerouted around $\f(X)$. As $\f$ is injective, $\g(X)$ is the only edge that
  is rerouted around $\f(X)$. Hence, there is only one more choice~for~$e'$: to
  be rerouted around the other endpoint $z$ of $e$. That is, $e'$ is critical
  and there is a $(k+1)$-crossing $X'\in\mathcal X\setminus\{X\}$ so that
  $e'=\g(X')$ and $\f(X')=z$.
  Since $e$ lies in $\R{X}$, and only the hook of $e'$ enters $\R{X}$,
  the edges $e$ and $e'$ cross only once.

  This proves (ii) and (iii), it remains to prove that $e$ is not part of a $(k+1)$-crossing.
  Recall that $e$ is crossed by at most $k-1$ edges in $X$ plus at most another edge $e'$ that
  has been rerouted around an endpoint $z=\f(X')$ of $e$. It follows that if $e$ is part of a $(k+1)$-crossing,
  then all $k-1$ edges in $X$ that cross $e$ (and that also pairwise cross) also cross $e'$.
  But $e'$ does not cross any of these $k-1$ edges in $G$ (because $X'$ and $X$ are disjoint by Property~\ref{pr:k-crossings-disjoint}), and in $G'$ it is rerouted around $z$, which is not an endpoint of any of these $k-1$ edges. This proves (i) and completes the proof of the lemma.
\end{proof}

\begin{lemma}\label{lem:noadjcross2}
  No two adjacent critical edges cross in $G'$.
\end{lemma}
\begin{proof}
  Suppose to the contrary that there are two critical edges $e_1=(u,v_1)$ and
  $e_2=(u,v_2)$ that cross in $G'$. Then by Lemma~\ref{lem:cross} one edge must
  have been rerouted around an endpoint of the other. As no edge is rerouted
  around its own endpoints, we may assume without loss of generality that $e_1$ has
  been rerouted around $v_2$. Then there exists a $(k+1)$-crossing $X\in \mathcal{X}$
  such that $e_1\in X$: namely $u$ is an endpoint of $e_1=\g(X)$
  and $v_2=\f(X)$ are vertices~of~$X$.
  Under these conditions, 
  $X$ is a home for $e_2$.
  It follows that $e_2$ is safe, which
  contradicts our assumption that it is critical.
\end{proof}

The next two lemmas describe properties related to edges that cross in $G'$ but not in $G$.

\begin{lemma}\label{lem:non-rerouted-3-crossings}
Every nonrerouted edge $e$ is crossed by at most three critical edges in~$G'$.  Further, if $e$ is crossed by exactly three critical edges, then two of them have been rerouted around distinct endpoints of $e$.
\end{lemma}
\begin{proof}
Since at most one edge has been rerouted around each vertex, by construction, it suffices to prove that there exists at most one critical edge crossing $e$ that has not been rerouted around an endpoint of $e$.

For this, note that any edge with this property crosses $e$ also in $G$, by Lemma~\ref{lem:cross}, and thus it belongs to the same $(k+1)$-crossing as~$e$. Since, by construction, at most one edge per $(k+1)$-crossing is critical, the statement follows.
\end{proof}

\begin{lemma}\label{lem:single-non-rerouted}
If $G'$ contains a $(k+1)$-crossing $X'$, then $X'$ contains at most one nonrerouted~edge.
\end{lemma}
\begin{proof}
Assume to the contrary that a $(k+1)$-crossing $X'$ in $G'$ contains at least two nonrerouted edges $e_1$ and $e_2$.
By Lemma~\ref{lem:nosafeedges}\eqref{nosafeedges-i}, every edge in $X'$ is critical or nonrerouted.

We first claim that there exists an edge $e_3\in \EE{X'}$ that does not cross $e_1$ in $G$.
If $e_1$ has fewer than $k$ crossings in $G$, then the claim follows by the pigeonhole principle.
If $e$ is part of a $(k+1)$-crossing $X$ in $G$, the claim follows from the fact that the edges of $X$ do not form a $(k+1)$-crossing in $G'$, due to a rerouting of one of its edges.
Finally, assume that $e_1$ has $k$ crossings in $G$ but it is not part of any $(k+1)$-crossing.
Then none of the edges crossing $e_1$ in $G$ can be part of a $(k+1)$-crossing in $G$, otherwise they each would have $k$ crossings in the $(k+1)$-crossing and an additional crossing with $e_1$, contradicting the $k$-planarity of $G$.
Hence none of the edges crossing $e_1$ in $G$ is critical, they all are nonrerouted, hence $X'$ is a $(k+1)$-crossing in $G$, contradicting our assumption that $e_1$ is not part of any $(k+1)$-crossing in $G$.
This completes the proof of~the~claim.

Note that $e_3$ is critical by Lemma~\ref{lem:nosafeedges}\eqref{nosafeedges-i}, which means that $e_3$ is part of a $(k+1)$-crossing in $G$ containing neither $e_1$ nor $e_2$. Hence, $e_2$ and $e_3$ do not cross in $G$ by Property~\ref{pr:k-crossings-disjoint}. We prove that they do not cross in $G'$, either, a contradiction to the assumption that $X'$ is a $(k+1)$-crossing. By Lemma~\ref{lem:cross}(b), all new crossings of $e_3$ are on its hook; however, since $e_1$ and $e_2$ cross in $G$ (which is simple), they do not share any endpoint, and the statement follows.
\end{proof}

\section{Proving $(k+1)$-quasiplanarity}\label{se:kquasiplanarity}

Denote by $G'$ the topological graph obtained from a \kplanarstg $G$ by executing a full rerouting operation (assuming that the functions $\f$ and $\g$ are injective).
In this section we first prove that, if $k \ge 3$, then $G'$ does not contain $(k+1)$-crossings.
If $k=2$, the current choice of $\f$ and $\g$ may not avoid the presence of $3$-crossings.
Thus, when $k=2$, we modify $\f$ and $\g$ to obtain quasiplanarity.

\begin{lemma}\label{lem:nobigcrossings}
Let $G$ be a \kplanarstg, where $k \ge 3$.
The topological graph $G'$ does not contain any $(k+1)$-crossing.
\end{lemma}
\begin{proof}
Assume for a contradiction that $G'$ contains a $(k+1)$-crossing $X'$. By Lemma~\ref{lem:nosafeedges}\eqref{nosafeedges-i}, $X'$ does not contain any safe edge, and by Lemma~\ref{lem:single-non-rerouted}, $X'$ contains at most one nonrerouted edge.

Suppose that $X'$ contains one such edge $e$. By Lemma~\ref{lem:non-rerouted-3-crossings}, there are at most three critical edges crossing $e$ in $G'$. If there are less than $3$, then the claim follows, as $k \geq 3$. If there are three, say $d$, $h$, and $l$, then we may assume by Lemma~\ref{lem:non-rerouted-3-crossings} that $d$ and $h$ have been rerouted around (distinct) endpoints of $e$. Thus, $d$ and $h$ do not cross in $G$, by Property~\ref{pr:k-crossings-disjoint}, as they belong to different $(k+1)$-crossings. Hence, they can cross in $G'$ only if one of them has been rerouted around an endpoint of the other, by Lemma~\ref{lem:cross}. This is impossible since neither $d$ nor $h$ shares an endpoint with $e$, as $G$ is simple.

Suppose that $X'$ contains only critical edges. Let $e$ be any edge of $X'$ and assume that is has been rerouted around vertex  $w$. Since at most one edge in $X'$ can be incident to $w$ by Lemma~\ref{lem:noadjcross2} and since $k \ge 3$, there are two edges in $X'$, say $d$ and $h$, that have been rerouted around distinct endpoints of $e$. As in the previous case, $d$ and $h$ do not cross.
\end{proof}

In the remainder of this section, we assume that $k=2$.
In this case,
Lemma~\ref{lem:nobigcrossings} does not hold, as some $3$-crossings may still appear after the global rerouting; see \figurename~\ref{fig:twin-fan} for examples. Next we characterize the $3$-crossings in $G'$. The characterization allows us to avoid these $3$-crossings by choosing suitable functions $\f$ and $\g$.

\begin{definition}
  Three edges $e_1,e_2,e_3\in E$ form a \emph{twin} configuration
  (\figurename~\ref{fig:twin}) in $G'$ if they are in two distinct $3$-crossings
  $X_1,X_2\in \mathcal X$, where $e_1=\g(X_1)$, $e_2=\g(X_2)$ and
  $e_3\in X_2\setminus\{e_2\}$, such that edge $e_1$ is incident to $\f(X_2)$,
  edge $e_3$ is incident to $\f(X_1)$ but not to $\f(X_2)$, and $e_3$ is drawn
  inside $\R{X_2}$.
\end{definition}

\begin{definition}
  Three edges $e_1,e_2,e_3\in E$ form a \emph{whirl} configuration
  (\figurename~\ref{fig:fan}) in $G'$ if they are in three pairwise distinct
  $3$-crossings $X_1,X_2,X_3\in \mathcal X$, where $e_1=\g(X_1)$, $e_2=\g(X_2)$,
  and $e_3=\g(X_3)$, such that edge $e_1$ is incident to $\f(X_2)$, edge $e_2$
  is incident to $\f(X_3)$, and edge $e_3$ is incident to $\f(X_1)$.
\end{definition}


\begin{figure}[hbtp]
  \centering%
  \begin{minipage}[b]{.38\textwidth}
    \includegraphics[width=\textwidth]{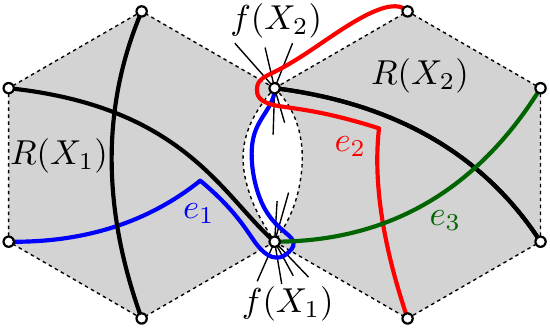}\subcaption{twin}\label{fig:twin}
  \end{minipage}
  \hfil
  \begin{minipage}[b]{.38\textwidth}
    \includegraphics[width=\textwidth]{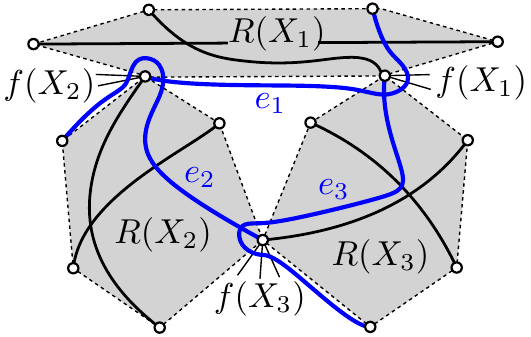}\subcaption{whirl}\label{fig:fan}
  \end{minipage}
  \caption{The global rerouting may produce 3-crossings in form of twins or whirls.\label{fig:twin-fan}}
\end{figure}

\begin{lemma}\label{lem:fan}
  Every $3$-crossing in $G'$ forms a twin or a whirl configuration.
\end{lemma}
\begin{proof}
  Let $e_1$, $e_2$, and $e_3$ be three edges that form a $3$-crossing in
  $G'$. By Lemma~\ref{lem:nosafeedges} we know that none of the three edges is
  safe. Hence, we may assume without loss of generality that $e_1$, $e_2$, and $e_3$ are either nonrerouted or critical. By Lemma~\ref{lem:single-non-rerouted}, at most one of them is nonrerouted, and thus at least two are critical. We assume that $e_1$ and $e_2$ are critical, and distinguish two cases, based on whether $e_3$ is critical or nonrerouted.\smallskip

We first consider the case in which $e_3$ is nonrerouted. Recall that every $3$-crossing reroutes at most
one critical edge in the global rerouting. For $i\in\{1,2\}$, let $X_i=\{c_i,d_i,e_i=\g(X_i)\}$ be the
$3$-crossing that triggered the rerouting of $e_i$ around the endpoint $\f(X_i)$ of
  $c_i$.

  Then by construction and Lemma~\ref{lem:cross} the only edges crossed by $e_i$, for $i\in\{1,2\}$,
  in $G'$ are $d_i$, edges incident to $\f(X_i)$, and at most two edges
  rerouted around an endpoint of $e_i$.

  Since $e_1$ and $e_2$ cross and are both critical, one of them has been rerouted
  around an endpoint of the other. Without loss of generality suppose that $e_2$ has been rerouted around an endpoint of $e_1$, that is, $e_1$ is incident to $\f(X_2)$.

  Since $e_3$ crosses both $e_1$ and $e_2$ and is nonrerouted, we have the following. On one hand, $e_3$ is either $d_1$ or incident to $\f(X_1)$; on the other hand, $e_3$ is either $d_2$ or incident to $\f(X_2)$. Thus, $e_3\in\{d_1,d_2,(\f(X_1),\f(X_2))\}$. We claim that $e_3=d_2$.

  To prove the claim, we first argue that $e_3\ne d_1$. By definition, $e_1$
  and $d_1$ do not share an endpoint, and $d_1$ is nonrerouted. The only
  nonrerouted edges that $e_2$ crosses in $G'$ are $d_2$ and edges incident
  to $\f(X_2)$. Since $\f(X_2)$ is an endpoint of $e_1$, it is not an endpoint of
  $d_1$. Therefore,~$e_2$ does not cross $d_1$, which implies $e_3\ne d_1$.

  It remains to prove that $e_3 \ne (\f(X_1),\f(X_2))$. Suppose, for a contradiction, that $e_3 = (\f(X_1),\f(X_2))$.
  This implies that the $3$-crossing $X_1$ is a home for $e_3$; namely, both $\f(X_1)$ and $\f(X_2)$ are vertices of $X_1$ (the former by definition and the latter as an endpoint of $e_1$), and $\f(X_2)$ is incident to $e_1$. However, this contradicts the assumption that $e_3$ is nonrerouted.
  Altogether it follows that $e_3=d_2$, as claimed, and so $e_1,e_2,e_3$ form a
  twin configuration.\smallskip

	We then consider the case in which $e_3$ is critical. Since only
  one edge of each $3$-crossing in~$\mathcal X$ is rerouted, $e_1,e_2,e_3$ come from
  pairwise distinct $3$-crossings $X_1,X_2,X_3$, with $e_i=\g(X_i)$, for
  $i\in\{1,2,3\}$. By Lemma~\ref{lem:cross} two of these edges cross if and only
  if one is rerouted around an endpoint of the other. By
  Lemma~\ref{lem:noadjcross2} the edges $e_1,e_2,e_3$ are spanned by six
  pairwise distinct endpoints. Therefore, every rerouting generates at most one
  crossing among $e_1,e_2,e_3$ and so every rerouting must generate a crossing
  between a different pair of segments. It follows that $e_1,e_2,e_3$ form a
  whirl configuration (with a suitable permutation of indices).
\end{proof}

%

In order to suitably select functions $\f$ and $\g$ so that $G'$ does not contain any twin or whirl configurations, we exploit an auxiliary conflict graph, which we define in the next subsection.

\subsection{Conflict digraph\label{ssec:conflict}}

We define a plane digraph $K$ on the same vertex set $V$ as $G$ that represents the interactions between the $3$-crossings in $\mathcal X$. The conflict digraph depends on $G$ and on the function $\f:\mathcal X\rightarrow V$, but it does not depend on the function $\g$. For every $3$-crossing $X\in\mathcal X$, we create five directed edges that are all directed towards $\f(X)$ and drawn inside $\R{X}$. These edges start from the five vertices on $\partial\R{X}$ other than $\f(X)$; see \figurename~\ref{fig:ctwin-cfan}. Note that two vertices in $V$ may be connected by two edges with opposite orientations lying in two different $3$-crossings (for instance, in a twin configuration as shown in \figurename~\ref{fig:ctwin}). However, $K$ contains neither loops nor parallel edges with the same orientation because $\f$ is injective and so every vertex can have incoming edges from at most one $3$-crossing.

\begin{figure}[hbtp]
  \centering%
  \begin{minipage}[b]{.4\textwidth}
    \includegraphics{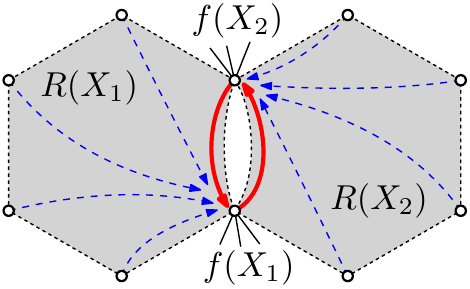}\subcaption{twin}\label{fig:ctwin}
  \end{minipage}
  \hfil
  \begin{minipage}[b]{.4\textwidth}
    \includegraphics{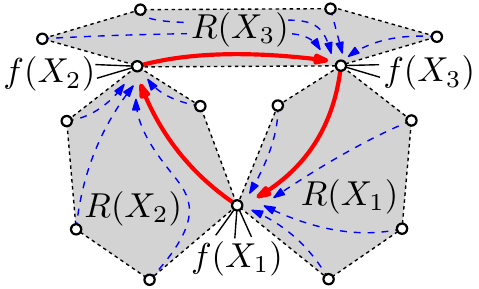}\subcaption{whirl}\label{fig:cfan}
  \end{minipage}
  \caption{Twin and whirl configurations induce cycles in the conflict graph.}\label{fig:ctwin-cfan}
\end{figure}

\begin{property}\label{pr:cycle}
The following properties hold for digraph $K$:
\begin{enumerate}[\rm (i)]
\item $K$ is a directed plane graph.
\item At every vertex $v\in V$, the incoming edges in $K$ are consecutive in the
  cyclic order of incident edges around $v$.
\item If $e_1=(v_1,v_2)$, $e_2=(v_2,v_3)$, and $e_3=(v_3,v_1)$ form a whirl configuration in $G'$, then $K$ contains a $3$-cycle $(v_1,v_2,v_3)$.
\item If $e_1=\g(X_1)$, $e_2=\g(X_2)$, and $e_3\in X_2$ form a twin configuration in $G'$, then the conflict digraph contains a $2$-cycle $(\f(X_1),\f(X_2))$.
\end{enumerate}
\end{property}
\begin{proof}
  (i) Each edge of $K$ lies in a region $\R{X}$, for some $X\in\mathcal
  X$. Since these regions are interior-disjoint, by
  Property~\ref{pr:k-crossings-disjoint}, edges from different regions do not
  cross. All edges in the same region $\R{X}$ are incident to $\f(X)$; so they
  can be drawn in the interior of $\R{X}$ without crossing each other. (ii) For
  each vertex $v\in V$, there is at most one $3$-crossing $X\in \mathcal X$ such
  that $v=\f(X)$, since $\f$ is injective. Since all incoming edges of $v$ lie
  in the region $\R{X}$, and all edges lying in $\R{X}$ are directed towards
  $v=\f(X)$, by construction, the statement follows. (iii--iv) Both claims
  follow directly from the definition of  twin and whirl configurations and the definition~of~$K$.
\end{proof}

\paragraph{Relations between cycles in $K$}
We observed that $K$ is a plane digraph, where every twin configuration induces
a 2-cycle and every whirl configuration induces a 3-cycle. So in order to prevent
the creation of twin and whirl configurations in $G'$, we need
to understand the structure of 2- and 3-cycles in the conflict digraph $K$. In
the following paragraphs we introduce some terminology and prove some structural
statements about cycles in $K$.

A cycle in the conflict digraph $K$ is \emph{short} if it has length two or
three. For a cycle $c$~in~$K$, let $\inte(c)$ denote the interior of $c$, let
$\exte(c)$ denote the exterior of $c$, let $\R{c}$ denote the compact
region bounded by $c$, and let $\V{c}$ denote the vertex set
of $c$. We use the notation $i \oplus 1:=1+(i\;\mathrm{mod}\;k)$ and
$i \ominus 1:=1+((k+i-2)\;\mathrm{mod}\;k)$ to denote successors and
predecessors, respectively, in a circular sequence of length $k$ that is indexed
$1,\ldots,k$. Let $c_1$ and $c_2$ be two cycles in the conflict graph $K$.  We
say that $c_1$ and $c_2$ are \emph{interior-disjoint} if
$\inte(c_1)\cap\inte(c_2)=\emptyset$. We say that $c_1$ \emph{contains} $c_2$
if $\R{c_2}\subseteq\R{c_1}$. In both cases, $c_1$ and $c_2$ may share vertices
and edges, but they may also be vertex-disjoint.  See \figurename~\ref{fig:contains}
for an example.

\begin{figure}[thbp]
  \centering
  \includegraphics[width=.39\textwidth]{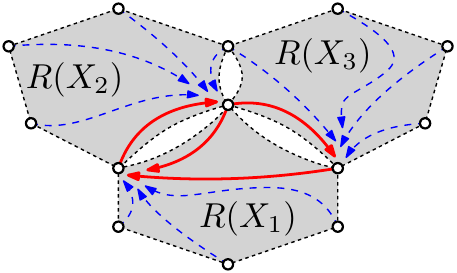}
  \caption{A (ghost) $3$-cycle that contains a $2$-cycle.\label{fig:contains}}
\end{figure}

\begin{lemma}\label{lem:V0}
  If a vertex $v\in V$ is incident to two interior-disjoint cycles in $K$, then
  these cycles have opposite orientations (clockwise vs.~counterclockwise).
  Consequently, every vertex $v\in V$ is incident to at most two
  interior-disjoint cycles in $K$.
\end{lemma}
\begin{proof}
  Let $v$ be incident to cycles $c_1$ and $c_2$ in $K$, and assume without loss
  of generality that $c_1$ is counterclockwise.  For $i\in \{1,2\}$, the cycle
  $c_i$ has an edge $e_i^{\rm in}$ directed into $v$ and an edge $e_i^{\rm out}$
  directed out of $v$ (possibly $e_1^{\rm in}=e_2^{\rm in}$ or
  $e_1^{\rm out}=e_2^{\rm out}$).

  By Property~\ref{pr:cycle}(ii), the edges directed to (resp., from) $v$ are consecutive in the
  rotation order of all edges incident to $v$. The edges $e_1^{\rm out}$ and $e_1^{\rm in}$
  (resp., $e_2^{\rm out}$ and $e_2^{\rm in}$) are also consecutive because the two cycles
  are interior-disjoint. It follows that the counterclockwise order of the four edges around
  $v$ is $(e_1^{\rm out}, e_1^{\rm in}, e_2^{\rm in}, e_2^{\rm out})$. So the
  cycle $c_2$ is clockwise, as required.
\end{proof}

\begin{lemma}\label{lem:cycdet}
  A short cycle in $K$ is uniquely determined by its vertex set.
\end{lemma}
\begin{proof}
  Recall that between any ordered pair $(u,v)$ of vertices there is at most one
  directed edge $(u,v)$ in $K$ because such an edge corresponds to a $3$-crossing
  $X \in\mathcal X$ with $u,v\in \V{X}$ and~$\f(X)=v$. As $\f$ is injective, there
  is at most one such $3$-crossing.

  So the statement is obvious for $2$-cycles. Consider two $3$-cycles $c_1$ and
  $c_2$ in $K$ with $\V{c_1}=\V{c_2}=\{v_1,v_2,v_3\}$. Without loss of
  generality, let $c_1=(v_1,v_2,v_3)$. If $c_1$ and $c_2$ share an edge, say
  $(v_1,v_2)$, then there is a unique way to complete this edge to a directed
  $3$-cycle $(v_1,v_2,v_3)=c_1=c_2$. Hence suppose that $c_1$ and $c_2$ are
  edge-disjoint, that is, $c_2=(v_3,v_2,v_1)$.

  Let $X_i$ denote the $3$-crossing with $f(X_i)=v_i$, for $i\in\{1,2,3\}$. All edges
  directed to $v_1$ are drawn inside $\R{X_1}$ between vertices of $\V{X_1}$, and
  both $(v_3,v_1)$, as an edge of $c_1$, and $(v_2,v_1)$, as an edge of $c_2$,
  are edges of $K$. Therefore, $v_1,v_2,v_3\in \V{X_1}$.  Symmetrically, it
  follows that $v_1,v_2,v_3\in \V{X_1}\cap\V{X_2}\cap\V{X_3}$. Three distinct
  $3$-crossings $X_1,X_2,X_3$ share three distinct vertices, contradicting
  Property~\ref{pr:kgon}\eqref{obs:kgon:2}. It follows that
  $c_2$ and $c_1$ have the same orientation and therefore $c_1=c_2$.
\end{proof}

\paragraph{Ghosts}
We say that a $3$-cycle in $K$ is a \emph{ghost} if two of its vertices induce
a $2$-cycle in $K$; see, e.g., \figurename~\ref{fig:contains}.
Let $\mathcal{C}$ denote the set of all short cycles in $K$ that are not ghosts.

\begin{lemma}\label{lem:max}
  Let $c_1,c_2\in\mathcal{C}$. If there is a vertex of $\V{c_1}$ in $\inte(c_2)$, then
  $c_2$ contains $c_1$.
\end{lemma}
\begin{proof}
  Suppose to the contrary that there exist short cycles $c_1,c_2\in\mathcal{C}$
  such that there is a vertex $v_1 \in \V{c_1}$ that lies in $\inte(c_2)$ but $c_2$ does not contain $c_1$.
  Then  some point along $c_1$ lies in $\exte(c_2)$. Since $K$ is a plane graph, an
  entire edge of $c_1$ must lie in $\exte(c_2)$. Denote this edge by
  $(v_2,v_3)$. Recall that $c_1$ is short (that is, it has at most three
  vertices), consequently, $c_1=(v_1,v_2,v_3)$.  Since $c_1$ has points in both
  $\inte(c_2)$ and $\exte(c_2)$, the two cycles intersect in at least two
  points. In a plane graph, the intersection of two cycles consists of vertices
  and edges. Consequently $\V{c_1}\cap\V{c_2}=\{v_2,v_3\}$.  Recall that $c_2$
  is also short, and so it has a directed edge between any two of its
  vertices. However, $(v_2,v_3)$ lies in $\exte(c_2)$, so the reverse edge
  $(v_3,v_2)$ is present~in~$c_2$. That is, $\{v_2,v_3\}$ induces a 2-cycle in $K$. Hence $c_1$ is a ghost, contrary to our assumption $c_1\in\mathcal{C}$.
\end{proof}

\paragraph{Smooth cycles}
Next we define a special type of cycles, called \emph{smooth}, so as to control
the interaction between cycles in $K$.

\begin{definition}
  Let $c=(v_1,\ldots,v_m)\in\mathcal{C}$.
  Recall that every edge in $K$ lies in a region $\R{X}$, $X\in\mathcal X$,
  and is directed towards $\f(X)$. So the cycle $c$ corresponds to a cycle of $3$-crossings
  $(X_1,\ldots,X_m)$, such that $v_i=\f(X_i)$ and $v_i$ lies in
  the common boundary $\partial \R{X_i}\cap \partial \R{X_{i\oplus 1}}$
  for $i=1,\ldots,m$. We say that the $3$-crossings
  $X_1,\ldots,X_m$ are \emph{associated} with $c$.
  The cycle $c$ is \emph{smooth} if none of the associated $3$-crossings has a
  vertex in $\inte(c)$. For example, the $3$-cycle in \figurename~\ref{fig:smooth}
  is smooth, but the one in \figurename~\ref{fig:nonsmooth} is not.
\end{definition}

\begin{figure}[hbtp]
  \centering%
  \begin{minipage}[b]{.39\textwidth}
    \includegraphics{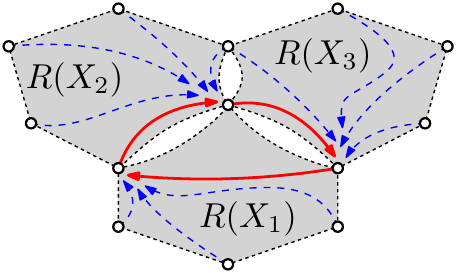}\subcaption{A smooth 3-cycle.}\label{fig:smooth}
  \end{minipage}
  \hfil
  \begin{minipage}[b]{.39\textwidth}
    \includegraphics{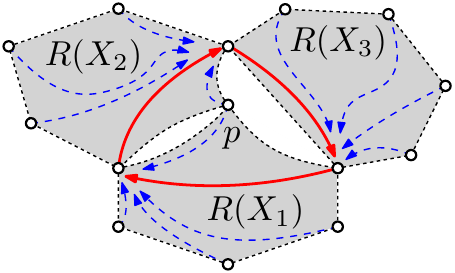}\subcaption{A nonsmooth 3-cycle.}\label{fig:nonsmooth}
  \end{minipage}
  \caption{Examples of smooth and nonsmooth cycles.\label{fig:6}}
\end{figure}

Note that a smooth cycle in $K$ may contain many vertices of various $3$-crossings in
its interior; the restrictions apply only to those (two or three) $3$-crossings that
are associated with the cycle. For instance, there might be many more $3$-crossings
in the white regions between the $3$-crossings in \figurename~\ref{fig:6}.

Let $\cs$ denote the set of all smooth cycles in $\mathcal{C}$, that is, the set
of all short nonghost cycles in $K$ that are smooth. In Section~\ref{ssec:f}, we show how
to choose $\f$ such that all cycles in $\mathcal{C}$ are smooth, that is,
$\mathcal{C}=\cs$.

\paragraph{Properties of smooth cycles}
The following three lemmas formulate important properties of smooth cycles
that hold for any injective function $\f$.

\begin{lemma}\label{lem:inbound}
  Let $c\in\cs$ and let $u$ be a vertex of $G$ that lies in $\inte(c)$. Then there is no edge
  $(u,v)$ in $K$ for any $v\in\V{c}$.
\end{lemma}
\begin{proof}
  Suppose for the sake of a contradiction that $(u,v)$ is an edge of $K$ with
  $v\in\V{c}$. Let $X$ be the $3$-crossing with $\f(X)=v$. Every edge directed into $v$
  is directed out of some other vertex in $\V{X}$, in particular, $u\in\V{X}$.
  As $X$ is associated with $c$, this contradicts the assumption that $c$ is smooth.
\end{proof}

\begin{lemma}\label{lem:containment}
  Let $c_1,c_2\in\cs$ so that $c_1\ne c_2$ and $c_2$ contains $c_1$. Then
  $\V{c_1}\cap\V{c_2}=\emptyset$.
\end{lemma}
\begin{proof}
  Suppose to the contrary that there exists a vertex $u\in\V{c_1}\cap\V{c_2}$.
  We claim that there is no vertex of $c_1$ that lies in $\inte(c_2)$.
  To see this, consider some $v\in\V{c_1}\cap\inte(c_2)$. Then following $c_1$ from $v$ to $u$ we
  find an edge $(x,y)$ of $K$ so that $x$ lies in $\inte(c_2)$ and
  $y\in\V{c_2}$. However, such an edge does not exist by
  Lemma~\ref{lem:inbound}. Hence there is no such $v\in\V{c_1}\cap\inte(c_2)$.
  Given that $c_2$ contains $c_1$, it follows  that $\V{c_1}\subseteq\V{c_2}$.

  If $c_1$ is a $3$-cycle, then the claim above implies that so is $c_2$,
  and Lemma~\ref{lem:cycdet} contradicts our assumption $c_1\ne c_2$.
  Hence $c_1$ is a $2$-cycle and $c_2$ is a $3$-cycle.
  But then $c_2$ is a ghost, in contradiction to $c_2\in\cs$.
\end{proof}

\begin{lemma}\label{lem:C0}
  Any two cycles in $\cs$ are interior-disjoint or vertex disjoint.
\end{lemma}
\begin{proof}
  Let $c_1,c_2\in \cs$ so that $c_1\neq c_2$. Suppose, to the contrary,
  that $\inte(c_1)\cap\inte(c_2)\ne\emptyset$ and $\V{c_1}\cap\V{c_2}\neq\emptyset$.
  Without loss of generality, an edge $(u_1,u_2)$ of $c_2$ lies in $\inte(c_1)$.

  We may assume that $u_1$ and $u_2$ are common vertices of $c_1$ and $c_2$.
  Indeed, if $u_1$ or $u_2$ were not common vertices of the cycles, then a
  vertex of $c_2$ would lie in the interior of $c_1$. Then $c_1$ contains $c_2$
  by Lemma~\ref{lem:max}, and $\V{c_1}\cap\V{c_2}=\emptyset$ by
  Lemma~\ref{lem:containment}.

  We may further assume that both $c_1$ and $c_2$ are 3-cycles. Indeed, if the
  vertex set of one of them contains that of the other, then one of them is a
  3-cycle and the other is a 2-cycle by Lemma~\ref{lem:cycdet}.
  Hence the $3$-cycle is ghost, contradicting the assumption that
  both $c_1$ and $c_2$ are present in $\mathcal{C}$, $\cs\subseteq \mathcal{C}$.

  Since $(u_1,u_2)$ is a directed edge of $c_2$ that lies in the interior of
  $c_1$, since $c_1$ is a 3-cycle that has an edge between any two of its vertices,
  and since $u_2$ can have incoming edges from at most one $3$-crossing (because $\f$ is injective), it follows that
  the edge $(u_2,u_1)$ is present in $c_1$. This implies that $c_3=(u_1,u_2)$
  is a 2-cycle in $K$. Therefore $c_3 \in \mathcal{C}$, and both
  $c_1$ and $c_2$ are ghost cycles in~$\cs\subseteq \mathcal{C}$,
  contradicting the definition of $\mathcal{C}$. Thus $c_1$ and $c_2$ are
  interior-disjoint or vertex disjoint,
  as claimed.
\end{proof}

\subsection{How to choose the function $\f$}\label{ssec:f}

As a next step, we show how to define the function $f$ so that in the resulting
conflict graph all nonghost cycles are smooth. The idea is to incrementally
modify the function $\f$ so that the number of vertices that are contained in
nonsmooth cycles decreases.

\begin{lemma}\label{lem:nested}
  Let $G=(V,E)$ be a \twoplanarstg, and let $\mathcal{X}$ be the set of $3$-crossings of $G$.
  There exists an injective function $\f: \mathcal X \rightarrow V$ such that
  every short cycle that is not a ghost in the conflict digraph $K$ of $G$ is smooth
  (that is, $\mathcal{C}=\cs$).
\end{lemma}
\begin{proof}
  Let $\f:\mathcal X \rightarrow V$ be an arbitrary injective function that maps
  every $3$-crossing $X\in\mathcal X$ to a vertex $v\in\V{X}$. Such a function
  exists by Lemma~\ref{lem:global}. We repeatedly modify the function $\f$ to
  achieve the desired property.

  Following the notation defined above, let $K$ be the conflict digraph of $G$
  determined by $\f$, and let $\mathcal C$ be the set of short cycles in $K$
  that are not ghosts.  Also, let $\cns\subseteq\mathcal{C}$ denote the subset
  of cycles in $\mathcal{C}$ that are not smooth. If $\cns=\emptyset$, then the
  proof is complete. As long as $\cns\ne\emptyset$, we repeatedly modify $\f$
  for some vertices in the region $\R{c}$ of a cycle $c\in \cns$. This
  modification of $\f$ correspondingly changes the conflict digraph $K$ (and hence
  the set $\cns$). As a measure of progress we maintain that the cardinality of
  the set $\vns$ decreases, where $\vns$ is the set of vertices that lie in the
  regions bounded by the cycles in $\cns$, that is,
  $\vns=V\cap(\bigcup_{c\in\cns}\R{c})$.

  A cycle $c\in\cns$ is \emph{maximal} if there exists no cycle $c'\in\cns\setminus\{c\}$ such that $c'$ contains $c$.  
  Recall that if a cycle $c\in\cns$ is not smooth, then there exists a vertex of some associated $3$-crossing that 
  lies in the interior of $c$. 
  
  \paragraph{One incremental modification of $\f$}
  Given an injective function $\f:\mathcal X \rightarrow V$ such that $\f(X)\in V(X)$ for every 
  $X\in\mathcal X$, a maximal cycle $c\in \cns$, an associated 3-crossing $X_1\in X$, 
  and a vertex $v\in\V{X_1}\cap\inte(c)$, we define a new function $\f'=F(f,c,X_1,v)$, 
  which is an injective function $\f':\mathcal X \rightarrow V$ such that $\f(X)\in V(X)$ for every
  $X\in\mathcal X$; and in particular $\f'(X_1)=v$. Later we will argue how to select 
  $c$, $X_1$, and $v$ more carefully so as to guarantee certain properties for $\f'$.
 
  Let $c=(v_1,\ldots,v_m)$ for $m\in\{2,3\}$ and let $X_1,\ldots,X_m$ denote the 
  associated $3$-crossings such that $v\in\V{X_1}\cap\inte(c)$.

  We define a new injective function $\f':\mathcal{X}\rightarrow V$ as follows.
  We set $\f'(X_1)=v$. For all 3-crossings $X\in \mathcal{X}$, $X\neq X_1$,
  for which $\R{X}\not\subseteq\R{c}$, we set $\f'(X)=\f(X)$. In particular,
  $\f(X_j)=\f(X_j)$ for $j=2,\ldots, m$ along the cycle $c$.
  For all remaining $X\in \mathcal{X}$, where $\R{X}\subseteq \R{c}$, we define
  $\f'(X)\in V\in \inte{c}$ using Hall's theorem. For these 3-crossings,
  Hall's condition is still satisfied by Lemma~\ref{lem:matching2} even if we 
  exclude up to 3 vertices along the cycle $c$, and we find an injective function
  $\f'$ as in Lemma~\ref{lem:global}. This completes the definition of $\f'=F(f,c,X_1,v)$.

  The modified function $\f'$ defines a new conflict digraph that we denote by $K'$.
  Note that both $K$ and $K'$ have the same vertex set, namely $V$.
  Since $\f(X_1)\neq \f'(X_1)=v$, the cycle $c$ of $K$ is not present in $K'$.
  Let $\mathcal{C}'$, $\cns'$ and $\vns'$ be defined analogously to
  $\mathcal{C}$, $\cns$ and $\vns$ in~$K'$. We claim that:

  \begin{enumerate}[(A)]
  \item\label{cl:A} every cycle in $\cns'\setminus\cns$ is contained in $c$, and
  \item\label{cl:B} $v_1\notin\R{c'}$ for any cycle $c'\in \cns'$.
  \end{enumerate}

  The combination of \eqref{cl:A} and \eqref{cl:B} immediately establishes
  $\vns'\subsetneq\vns$, our measure of progress. We call a cycle $b$ \emph{bad}
  if it violates \eqref{cl:A}, that is, $b\in\cns'\setminus\cns$ and $c$ does
  not contain $b$.

  We first show that \eqref{cl:A} implies \eqref{cl:B}. Note that $v_1$ is a
  vertex of every cycle $d\in\cns$ for which $v_1\in\R{d}$. To see this, let
  $d\in\cns$ with $v_1\in\R{d}$. If $v_1\in\inte(d)$, then $d$ contains $c$ by
  Lemma~\ref{lem:max}, and so $d=c$ by the maximality of $c$. Hence
  $v_1\in \V{d}$, as claimed. As $v_1$ has no incoming edge in $K'$, it follows
  that setting $\f(X_1)=v$ destroys \emph{all} cycles in $\cns$ that contain
  $v_1$. Therefore, if there is a cycle $b\in\cns'$ for which $v_1\in\R{b}$,
  then $b$ is a new cycle, that is, $b\in \cns'\setminus\cns$.  In fact, in order to
  contain $v_1$, the cycle $b$ must be bad: If $v_1\in\R{b}\subseteq\R{c}$, then
  $v_1$ is a vertex of $b$, which is impossible because $v_1$ has no incoming
  edge in $K'$.  By \eqref{cl:A}, there is no bad cycle, and so no cycle in
  $\cns'$ contains $v_1$ and \eqref{cl:B} holds, as claimed.

  To prove \eqref{cl:A}, we distinguish some cases and argue separately in each case.
  Before the case distinction, we give a common characterization of bad cycles.

   Recall that we do not change $\f(X_i)$, for $i\in\{2,\ldots,m\}$, i.e., $\f(X_i)=\f'(X_i)$, for $i\in\{2,\ldots,m\}$.
   Therefore, the conflict digraph has the same edges
  inside $\R{X_i}$, for $i\in\{2,\ldots,m\}$, before and after the modification of $\f$. In particular,
  as $K'$ is plane, its edges can only cross the edge $(v_m,v_1)$ of $c$
  (because it is in $K$ but not in $K'$),
  which lies in $\R{X_1}$. All edges of $K'$ inside~$\R{X_1}$ are
  directed to vertex $v\in \inte(c)$. Therefore, every edge in $K'$ that crosses
  $(v_m,v_1)$ has one endpoint in $\inte(c)$ and one endpoint in $\exte(c)$.

\begin{figure}[hbtp]
  \centering%
  \begin{minipage}[b]{.39\textwidth}
    \includegraphics{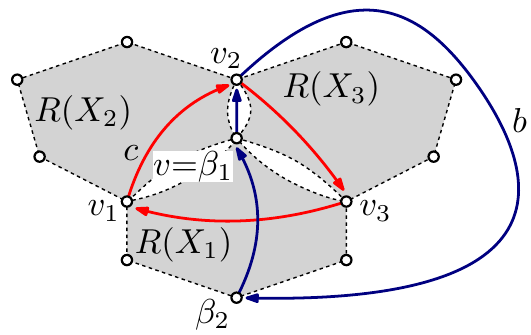}\subcaption{A bad cycle.}\label{fig:bad-cycle}
  \end{minipage}
  \hfil
  \begin{minipage}[b]{.39\textwidth}
    \includegraphics[page=2]{badcycle-new}\subcaption{Case 1.}\label{fig:case1}
  \end{minipage}
  \\
  \begin{minipage}[b]{.39\textwidth}
    \includegraphics[page=3]{badcycle-new}\subcaption{Case 2.1.}\label{fig:case2.1}
  \end{minipage}
    \hfil
    \begin{minipage}[b]{.39\textwidth}
    \includegraphics[page=4]{badcycle-new}\subcaption{Case 2.2.}\label{fig:case2.2}
  \end{minipage}
  \caption{Illustrations for the proof of Lemma~\ref{lem:nested}.}
\end{figure}

  \paragraph{Characterization of bad cycles}
  Let $b$ denote a bad cycle in $K'$; refer to Fig.~\ref{fig:bad-cycle}. Note that $\f'$ is unchanged with respect to $\f$ for $3$-crossings in the exterior of $c$.
  Therefore, every cycle in $K'$ that involves only vertices in $c\cup\exte(c)$ is also a cycle in $K$. This implies that every new cycle in
  $\cns'\setminus\cns$ must have a vertex in $\inte(c)$,
  and so  $b$ has a vertex $\beta_1\in\inte(c)$.

  We claim that $b$ also has a vertex $\beta_2\in\exte(c)$.  To see this,
  suppose to the contrary that $\V{b}\subset\R{c}$.  Since $b$ is bad, $c$ does
  not contain $b$, that is, $\R{b}\not\subset\R{c}$. Hence $b$ has an edge
  $(u_1,u_2)$ that passes through $\exte(c)$. As noted above, every edge of $K'$
  that crosses $(v_m,v_1)$ has a vertex in $\exte(c)$. Therefore, $b$ does not
  cross $(v_m,v_1)$ and so $u_1,u_2\in \V{b}\cap \V{c}$ and
  $b=(u_1,u_2,\beta_1)$. As $v_1$ has no incoming edge in $K'$, we have
  $v_1\notin \V{b}$ and so $m=3$ and $\{u_1,u_2\}=\{v_2,v_3\}$.  Since $c$ is
  short, it contains an edge between $u_1$ and $u_2$, and as the edge
  $(u_1,u_2)$ passes through $\exte(c)$, the reverse edge $(u_2,u_1)$ is an edge
  of $c$. But then $\{u_1,u_2\}=\{v_2,v_3\}$ induce a $2$-cycle in $K'$, and $b$
  is a ghost, contradicting our assumption that $b\in \cns'$.  This proves the
  claim that $b$ has a vertex $\beta_2\in\exte(c)$.

  Given the position of $\beta_1$ and $\beta_2$, it follows that $b$ crosses~$c$. 
  As noted above, $(v_m,v_1)$ is the only edge of $c$ that can be crossed
  by an edge of $K'$. This leaves only four options for $b$ to cross~$c$:
  the at most three vertices of $c$ and the edge $(v_m,v_1)$. As all edges of $K'$
  that cross $(v_m,v_1)$ are directed towards $v$, the cycle $b$ crosses the edge
  $(v_m,v_1)$ of $c$ at most once. Moreover, if $b$ crosses $(v_m,v_1)$, then
  the crossing edge starts from a vertex of $X_1$ and goes to the vertex $v$.
  As~$b$ has at most three vertices and due to the position of $\beta_1$ and $\beta_2$,
  the cycles $b$ and $c$ can share at most one vertex.

  Altogether it follows that a bad cycle $b$ has exactly three vertices: the
  vertex $\beta_1=v\in\inte(c)$, a vertex $\beta_2\in\exte(c)\cap\V{X_1}$ that
  precedes $v$ in $b$, and a third vertex $\beta_3\in\{v_2,v_m\}$. Hence
  $b=(v,v_j,\beta_2)$, for some $j\in\{2,m\}$. (We cannot have $\beta_3=v_1$
  because $v_1\notin \f'(\mathcal X)$ after the reassignment.)  Note that
  $\beta_3=v_j$, for $j\in\{2,m\}$, implies that $v$ is a vertex of $X_j$ because $(v,v_j)$ is an edge of $b$.

  \paragraph{Case analysis}
  In order to prove \eqref{cl:A}, we distinguish three cases: Case 1, Case 2.1,
  and Case 2.2 below.

  \paragraph{Case~1: There is a maximal cycle $c\in\cns$ and a vertex
    $v\in \inte(c)$ such that $v$ is incident to exactly one of
    $X_1,\ldots , X_m$; refer to Fig.~\ref{fig:case1}} We may assume that $v$ is incident to $X_1$ (by
  cyclically relabeling $X_1,\ldots, X_m$ if necessary). We set $\f'=F(f,c,X_1,v)$. 
  By the discussion above, the cycle $c$ is destroyed and no bad cycle is created
  (because the existence of a bad cycle implies that $v$ is also a vertex of at
  least one of the other $3$-crossing(s) $X_j$, for $j\in\{2,m\}$).

  \paragraph{Case~2: For every maximal cycle $c\in \cns$, every vertex of
    $X_1,\ldots , X_m$ in $\inte(c)$ is incident to at least two $3$-crossings
    in $\{X_1,\ldots , X_m\}$} We consider two subcases.

  \paragraph{Case~2.1: There are two interior-disjoint maximal $3$-cycles in
    $\cns$ that share an edge; refer to Fig.~\ref{fig:case2.1}} Denote these two cycles by $c_1=(v_1,v_2,v_3)$
  and $c_2=(v_1,v_4,v_3)$, and let $X_1,\ldots,X_4$ denote the associated
  $3$-crossings so that $v_i=\f(X_i)$, for $i\in \{1,2,3,4\}$. Note that
  $v_2\notin \V{X_1}$ because then $v_1$ and $v_2$ would induce a $2$-cycle in
  $K$, which contradicts the fact that $c_1$ is not a ghost. Analogously, it
  follows that $v_4\notin \V{X_1}$. The union of the edges $(v_1,v_2)$,
  $(v_2,v_3)$, $(v_1,v_4)$, and $(v_4,v_3)$ forms an (undirected) closed Jordan
  curve $\hat{c}$. On one hand, none of the four edges that form $\hat{c}$ is
  oriented towards $v_1=\f'(X_1)$, and so the curve $\hat{c}$ lies in the
  exterior of $\R{X_1}$. On the other hand, the (closed) region $\R{\hat{c}}$
  bounded by $\hat{c}$ contains the edge $(v_3,v_1)$ in $\R{X_1}$. It follows that
  $\R{\hat{c}}\supset \R{X_1}$. Consequently, all four vertices in
  $\V{X_1}\setminus\{v_1,v_3\}$ lie in $\inte(c_1)\cup\inte(c_2)$. Without loss
  of generality, we may assume that at least two vertices of $\V{X_1}$ lie in
  $\inte(c_1)$. By Property~\ref{pr:kgon}\eqref{obs:kgon:2}, at most one vertex 
  in $\inte(c_1)$ is incident to all of $X_1,X_2,X_3$, there exists a vertex 
  $v\in \V{X_1}\cap \inte(c_1)$ incident to either~$X_2$ or $X_3$ (but not both).

  We select $c=c_1$ and set $\f'=F(f,c,X_1,v)$. 
  As noted above, any bad cycle is of the form $b=(v,v_j,\beta_2)$, 
  where $j\in \{2,3\}$ and $\beta_2\in\exte(c_1)\cap\V{X_1}$.

  If $j=2$, we have $v_2\in\exte(c_2)$ and
  $\beta_2\in\exte(c_1)\cap\V{X_1}\subset \inte(c_2)\subset\inte(\hat{c})$. In
  particular, the edge $(v_2,\beta_2)$ of $b$ crosses~$c_2$, in contradiction to
  the fact that every edge in $K'$ that crosses~$c_2$ crosses the edge
  $(v_3,v_1)$ of $c_2$ and therefore goes to $v$.

  Otherwise, $j=3$ and the edge $(v,v_3)$ of $b$ together with $v_3\in\V{X_1}$
  and therefore an edge $(v_3,v)$ in $K'$ makes $b$ a ghost, in contradiction to
  $b\in\cns'$.

  Therefore, in either case both maximal cycles $c_1$ and $c_2$ are destroyed,
  and no bad cycle is created.

  \paragraph{Case~2.2: There are no two interior-disjoint maximal $3$-cycles in
  $\cns$ that share an edge; refer to Fig.~\ref{fig:case2.2}} Let $c\in\cns$ be an arbitrary maximal cycle, where 
  $c=(v_1,\ldots,v_m)$ for $m\in\{2,3\}$, and let $X_1$ be an arbitrary associated 
  3-crossing for which there exists a vertex $v\in\V{X_1}\cap\inte(c)$.
  We set $\f'=F(f,c,X_1,v)$. 
  
  Assume first that $m=2$.
  As noted above, any bad cycle $b$ has exactly
  three vertices: $b=(v,v_2,\beta_2)$, where $v=\f'(X_1)$, $v_2=\f'(X_2)$, and
  $\beta_2=\f'(X_3)$ for some $3$-crossing $X_3$ in the exterior of $c$. Note that, by the condition of Case~2 we know that
  $v\in\V{X_1}\cap\V{X_2}$.
  Further, $v,v_2\in \V{X_1}\cap\V{X_2}$ implies that $(v,v_2)$ is a $2$-cycle in $K'$.
  That is, $b$ is a ghost in $K'$, in contradiction to $b\in\cns'$.

  Assume next that $m=3$ (as depicted in Fig.~\ref{fig:case2.2}).
  By the condition of Case~2 we may assume (by cyclically relabeling $(X_1,X_2,X_3)$ if
  necessary) that $v\in\V{X_1}\cap\V{X_2}$ (and possibly, $v\in\V{X_3}$).
  As noted above, any bad cycle $b$ has exactly three vertices:
  $b=(v,v_j,\beta_2)$, where $v=\f'(X_1)$, $v_j=\f'(X_j)$ for $j\in \{2,3\}$, and
  $\beta_2=\f'(X_4)$ for some $3$-crossing $X_4$ in the exterior of $c$. Assume that $b$
  is maximal with these properties.

  If $j=2$, then $c'=(v_1,v_2,\beta_2)$ is a maximal $3$-cycle in the original
  conflict digraph $K$.  We claim that the cycle $c'$ does not contain
  $c$. Suppose to the contrary that $c'$ contains $c$. Then
  $v,v_3\in\inte(c')$. Hence $c'$ is not smooth, contradicting our assumption
  that $c\in\cns$ is maximal. This proves the claim. It follows that $c$ and
  $c'$ are interior-disjoint maximal cycles in $\cns$ that share the edge
  $(v_1,v_2)$, contradicting our assumption in Case~2.2.

  Otherwise, $j=3$ and then $v$ is incident to $X_3$. In this case,
  $v,v_3\in \V{X_1}\cap \V{X_3}$, and we create a $2$-cycle $(v,v_3)$ in $K'$.
  Hence $b=(v,v_3,\beta_2)$ is a ghost, contradicting our assumption
  $b\in \cns'$.

  Consequently, there are no bad cycles when $m=3$ and $j\in \{2,3\}$.

\medskip
  In all three cases, we have shown that no bad cycle is created,
  which confirms \eqref{cl:A}. By \eqref{cl:A} and \eqref{cl:B}, each
  incremental modification of the initial function $\f$ strictly decreases the set $\vns$.
  After at most $|V|$ repetitions, we obtain an injective function
  for which $\cns=\emptyset$, as required.
\end{proof}

\subsection{How to choose the function $\g$}

Let $\f:\mathcal X \rightarrow V$ be a function such that $\mathcal{C}=\cs$
(that is, all short nonghost cycles in the corresponding conflict digraph $K$
are smooth), which exists by Lemma~\ref{lem:nested}.  As a first step, we will
use Hall's theorem to show that there is a matching of the cycles in
$\mathcal{C}$ to the vertices in $V$ such that every cycle $c\in\mathcal{C}$ is
matched to an incident vertex $s(c)$. Then, our plan is to break the cycle $c$
at the 3-crossing $X\in\mathcal{X}$ for which $\f(X)=s(c)$, by choosing the
value of $\g(X)$ appropriately.

For a subset $\mathcal{B}\subseteq\mathcal{C}$, let $\V{\mathcal{B}}$ denote the
set of all vertices incident to some cycle in $\mathcal{B}$.

\begin{lemma}\label{lem:disjoint}
  For every set $\mathcal{B}_0\subseteq \mathcal{C}$ of pairwise
  interior-disjoint cycles, $|\mathcal{B}_0| \le |\V{\mathcal{B}_0}|$.
\end{lemma}
\begin{proof}
  We use double counting. Let $I$ be the set of all pairs
  $(v,c)\in V\times\mathcal{B}_0$ such that $v$ is incident to $c$. Every cycle
  is incident to at least two vertices, hence $|I|\geq 2|\mathcal{B}_0|$. By
  Lemma~\ref{lem:V0}, every vertex is incident to at most two interior-disjoint
  cycles. Consequently, $|I|\leq 2|\V{\mathcal{B}_0}|$.  The combination of the
  upper and lower bounds for $|I|$ yields
  $|\mathcal{B}_0|\leq |\V{\mathcal{B}_0}|$, as claimed.
\end{proof}

\begin{lemma}\label{lem:Hall}
  For every set $\mathcal{B}\subseteq \mathcal{C}$ of cycles, we have
  $|\mathcal{B}|\leq |\V{\mathcal{B}}|$.
\end{lemma}
\begin{proof}
  We proceed by induction on the number of cycles in $\mathcal{B}$.  In the base
  case, we have one cycle, which has at least two vertices.

  Assume $|\mathcal{B}|\geq 2$, and let $\mathcal{B}_0\subseteq\mathcal{B}$ be
  the set of cycles in $\mathcal{B}$ that are maximal for containment.  By
  Lemma~\ref{lem:C0} the cycles in $\mathcal{B}_0$ are pairwise
  interior-disjoint, and by Lemma~\ref{lem:disjoint}, we have
  $|\mathcal{B}_0|\leq |\V{\mathcal{B}_0}|$.
  Induction for $\mathcal{B}\setminus \mathcal{B}_0$ yields
  $|\mathcal{B}\setminus \mathcal{B}_0|\leq |\V{\mathcal{B}\setminus
    \mathcal{B}_0}|$. By Lemma~\ref{lem:C0}, the vertex sets $\V{\mathcal{B}_0}$
  and $\V{\mathcal{B}\setminus \mathcal{B}_0}$ are disjoint. The combination of
  the two inequalities yields $|\mathcal{B}|\leq |\V{\mathcal{B}}|$.
\end{proof}

\begin{lemma}\label{lem:s}
  There exists an injective function $s:\mathcal{C}\rightarrow V$ that
  maps every cycle in $\mathcal{C}$ to one of its vertices.
\end{lemma}
\begin{proof}
  Consider the bipartite graph with partite sets $\mathcal{C}$ and $V$, where
  the edges represent vertex-cycle incidences. By Hall's theorem and
  Lemma~\ref{lem:Hall} (Hall's condition), there exists a matching of
  $\mathcal{C}$ into $V$, in which each cycle in $\mathcal{C}$ is matched to an
  incident vertex.
\end{proof}

We are ready to define the function $\g:\mathcal X \rightarrow E$, that maps
every $3$-crossing $X\in \mathcal X$ to one of its edges.
\begin{lemma}\label{lem:config}
Let $\f:\mathcal X \rightarrow V$ be a function obtained by Lemma~\ref{lem:nested},
and let $K$ be the corresponding conflict digraph.
There is a function $\g:\mathcal X \rightarrow E$ such that
  \begin{itemize}\itemsep -2pt
  \item for every $X\in \mathcal X$, $\g(X)\in X$ and $\g(X)$ is not incident to
    $\f(X)$;
  \item for every $2$-cycle $(\f(X_1),\f(X_2))$ in $K$, the edges $\g(X_1)$ and
    $\g(X_2)$ do not cross in $G'$;
  \item for every $3$-cycle $(\f(X_1)$, $\f(X_2)$, $\f(X_3))$ in $K$, at least two of the
    edges in $\{\g(X_1)$, $\g(X_2)$, $\g(X_3)\}$ do not cross in $G'$.
  \end{itemize}
\end{lemma}
\begin{proof}
  By Lemma~\ref{lem:s}, there is an injective function
  $s:\mathcal{C}\rightarrow V$ that maps every cycle $c\in \mathcal{C}$ to one
  of its vertices. For each cycle $c\in \mathcal{C}$, vertex $s(c)$ is the
  endpoint of some directed edge $(q(c), s(c))$ in $K$.
  Consequently, there is a $3$-crossing $X\in \mathcal X$ such that $s(c)=\f(X)$ and
  $q(c)\in\V{X}$. We say that the $3$-crossing $X$ is \emph{assigned} to the cycle $c$.
  We define $\g:\mathcal{X}\rightarrow E$ by successively selecting $\g(X)\in X$
  for every $3$-crossing $X\in \mathcal{X}$. We distinguish between two types of
  $3$-crossings, depending on whether or not they are assigned to a $2$-cycle of $\mathcal{C}$.

  \paragraph{$3$-crossings that are not assigned to a $2$-cycle} For every
  $3$-crossing $X$ that is not assigned to any cycle, choose $\g(X)$ to be an
  arbitrary edge in $X$ that is not incident to the vertex $\f(X)$. For every
  $3$-crossing $X$ that is assigned to a $3$-cycle $c\in \mathcal{C}$, choose
  $\g(X)$ to be the (unique) edge in $X$ that is incident to neither $q(c)$ nor
  $s(c)$. If $c=(\f(X_1),\f(X_2),\f(X_3))$ and without loss of generality
  $s(c)=\f(X_2)$, then $\g(X_2)$ is not incident to $\f(X_1)=q(c)$, consequently
  $\g(X_1)$ is disjoint from $\g(X_2)$ in $G'$. (Note that $\g(X_1)$ is not incident to
  $\f(X_2)=s(c)$ because this would induce a 2-cycle $(f(X_1),f(X_2))$ in $K$,
  making $c$ a ghost.)

  \paragraph{$3$-crossings assigned to $2$-cycles} Consider a $2$-cycle
  $c\in \mathcal{C}$, and let $X_1$ and $X_2$ denote the associated
  $3$-crossings so that without loss of generality $s(c)=\f(X_1)$. Assume
  without loss of generality that $c$ is oriented clockwise. We distinguish
  three cases.

  \paragraph{Case~1: $\g(X_2)$ has already been selected and $\g(X_2)$ is
    incident to $\f(X_1)$} Then let $\g(X_1)$ be the unique edge in $X_1$
  incident to $\f(X_2)$ (see \figurename~\ref{fig:twin:a}). We claim that
  $\g(X_1)$ and $\g(X_2)$ do not cross in $G'$. As both edges are rerouted, by
  Lemma~\ref{lem:cross} they can only cross in the neighborhood of $\f(X_1)$ or
  $\f(X_2)$. Let $a_i$ be the edge of $X_i$ incident to $\f(X_i)$, for
  $i\in\{1,2\}$. The edge $\g(X_i)$, for $i\in\{1,2\}$, follows $a_i$ towards
  the neighborhood of $\f(X_i)$ and then crosses the edges incident to $\f(X_i)$
  following $a_i$ in clockwise order (the orientation of $c$) until reaching the
  edge $(\f(X_1),\f(X_2))$. Then $\g(X_i)$ follows $\edge{\f(X_1)}{\f(X_2)}$ to its
  other endpoint, without crossing the edge. Therefore, the path formed by the
  edges $a_1$, $\edge{\f(X_1)}{\f(X_2)}$, and $a_2$ splits the neighborhoods of
  $\f(X_1)$ and $\f(X_2)$ into two components so that $\g(X_1)$ and $\g(X_2)$
  are in different components. Thus $\g(X_1)$ and $\g(X_2)$ do not cross, as
  claimed.

  \begin{figure}[htbp]
    \centering%
    \begin{minipage}[b]{.48\textwidth}
    \centering
    	\includegraphics[scale=1]{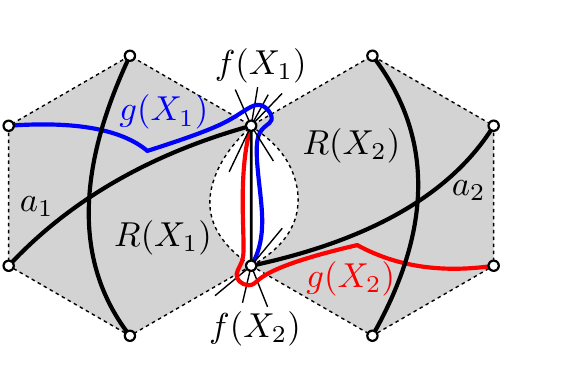}\subcaption{}\label{fig:t:a1}
  	\end{minipage}
  	\hfil
  	\begin{minipage}[b]{.48\textwidth}
  	\centering
    	\includegraphics[scale=1]{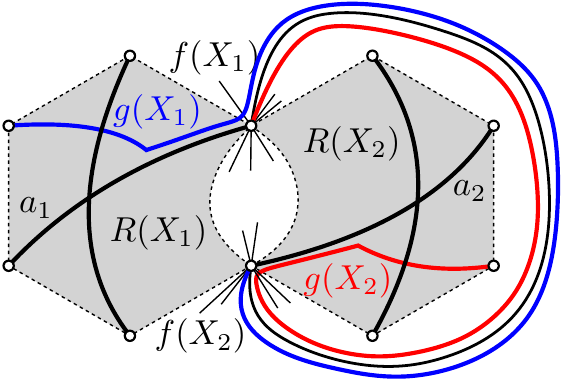}\subcaption{}\label{fig:t:a2}
  	\end{minipage}\hfil
    \caption{In Case~1, the edge $\g(X_2)$ is incident to $\f(X_1)$. We set
      $\g(X_1)$ so that it is incident to $\f(X_2)$. The edge $\edge{\f(X_1)}{\f(X_2)}$ may
      be drawn in various ways in $G$, two examples are shown above. Regardless
      of how $\edge{\f(X_1)}{\f(X_2)}$ is drawn, the edge separates $\g(X_1)$ and $\g(X_2)$
      and ensures that they are disjoint.\label{fig:twin:a}}
  \end{figure}

  \paragraph{Case~2: $\g(X_2)$ has already been selected and $\g(X_2)$ is not
    incident to $\f(X_1)$} Then let $\g(X_1)$ be the unique edge in $X_1$
  incident to neither $\f(X_1)$ nor $\f(X_2)$ (see
  \figurename~\ref{fig:t:d}). We claim that $\g(X_1)$ and $\g(X_2)$ do not cross
  in $G'$. As both edges are rerouted, by Lemma~\ref{lem:cross} they can only
  cross in the neighborhood of $\f(X_1)$ or $\f(X_2)$. But as $\g(X_1)$ is not
  incident to $\f(X_2)$, there is a neighborhood of $\f(X_2)$ that is disjoint
  from $\g(X_1)$, and so $\g(X_1)$ and $\g(X_2)$ do not cross there. Similarly,
  there is a neighborhood of $\f(X_1)$ that is disjoint from $\g(X_2)$, and so
  $\g(X_1)$ and $\g(X_2)$ do not cross there, either. It follows that $\g(X_1)$
  and $\g(X_2)$ do not cross in $G'$, as claimed.

  \begin{figure}[htbp]
    \begin{center}
    \begin{minipage}[b]{.48\textwidth}
      \centering
      \includegraphics[scale=1]{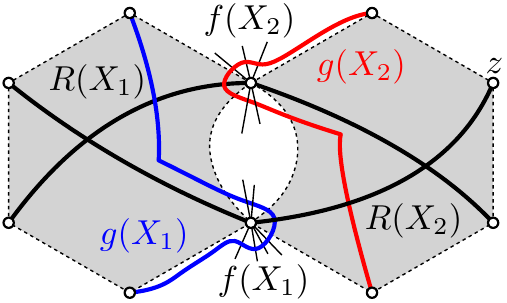}\subcaption{}\label{fig:t:d}
    \end{minipage}
    \hfil
    \begin{minipage}[b]{.48\textwidth}
    \centering
      \includegraphics[scale=1]{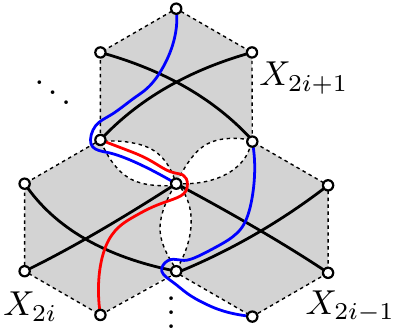}\subcaption{}\label{fig:t:c}
    \end{minipage}
    \end{center}
    \caption{(a)~In Case~2, the edge $\g(X_2)$ is not incident to $\f(X_1)$. We
      set $\g(X_1)$ so that it is not incident to $\f(X_2)$, to ensure that
      $\g(X_1)$ and $\g(X_2)$ are disjoint. (b)~In Case~3 we face a cycle of
      2-cycles. We consistently select edges to be rerouted in even (red edge)
      and odd (blue edge) $3$-crossings so that they are pairwise
      disjoint~in~$G'$.\label{fig:twin:b}}
  \end{figure}

  \paragraph{Case~3: no $3$-crossing $X_1$ is assigned to a $2$-cycle so that
    $\g(X_2)$ has already been selected} Then we are left with $3$-crossings
  that correspond to $2$-cycles and form \emph{cycles} $L=(X_1,\ldots,X_\ell)$
  such that $(\f(X_i),\f(X_{i\oplus 1}))$ is a $2$-cycle in $\mathcal{C}$, for
  $i=1,\ldots,\ell$. These cycles are interior-disjoint by Lemma~\ref{lem:C0},
  and any two consecutive cycles in $L$ have opposite orientations by
  Lemma~\ref{lem:V0}. It follows that $\ell$ is even.

  Since every 2-cycle in $L$ is smooth, the three vertices $\f(X_{i\ominus 1})$,
  $\f(X_i)$, and $\f(v_{i\oplus 1})$ are consecutive along
  $\partial\R{X_i}$. For every odd $i\in\{1,\ldots,\ell\}$, let $\g(X_i)$ be the
  (unique) edge in $X_i$ incident to $\f(X_{i\ominus 1})$ (and incident to
  neither $\f(X_i)$ nor $\f(X_{i\oplus 1})$). Similarly, for every even
  $i\in\{1,\ldots,\ell\}$, let $\g(X_i)$ be the edge in $X_i$ incident to
  $\f(X_{i\oplus 1})$ (and incident to neither $\f(X_i)$ nor
  $\f(X_{i\ominus 1})$). Refer to \figurename~\ref{fig:t:c} for an illustration.

  For every odd index $i\in\{1,\ldots,\ell\}$, the rerouted edges $\g(X_i)$ and
  $\g(X_{i\oplus 1})$ are incident to neither $\f(X_{i\oplus 1})$ nor
  $\f(X_i)$. Similarly, for every even index $i\in\{1,\ldots, \ell\}$, the
  rerouted edges $\g(X_i)$ and $\g(X_{i\oplus 1})$ are incident to
  $\f(X_{i\oplus 1})$ and $\f(X_i)$, respectively. In both cases, the rerouted
  edges $\g(X_i)$ and $\g(X_{i\oplus 1})$ are disjoint.

  \paragraph{Ghost cycles} It remains to consider ghost cycles. Let $c_1$ be a
  ghost cycle in $K$. Without loss of generality, assume that
  $c_1=(v_1,v_2,v_3)$, where $v_1=\f(X_1)$, $v_2=\f(X_2)$, $v_3=\f(X_3)$, and
  $c_2=(v_1,v_2)$ is a $2$-cycle in $\mathcal{C}$. Recall that $c_2$ is smooth
  (cf.~Lemma~\ref{lem:nested}). By construction, $\g(X_1)$ and $\g(X_2)$ do not
  cross in $G'$.  Hence at least two of the edges in
  $\{\g(X_1), \g(X_2), \g(X_3)\}$ do not cross in $G'$, as required.
\end{proof}

\subsection{Putting all the results together}

The results in this section can be used to prove that every \twoplanarstg can be transformed into a \quasi{3} topological graph by means of a suitable global rerouting operation.

\begin{lemma}\label{lem:g}
  Let $G=(V,E)$ be a \twoplanarstg, and let $\mathcal X$
  be the set of its $3$-crossings.
  There exist functions $\f:\mathcal X \rightarrow V$ and
  $\g:\mathcal X \rightarrow E$ such that the topological graph $G'$, $G' \simeq G$,
  obtained from $G$ by applying a full rerouting operation with functions $(\g,\f)$ is \quasi{3}.
\end{lemma}
\begin{proof}
 Lemmas~\ref{lem:nested} and~\ref{lem:config} imply the existence of two functions $\f$ and $\g$
 such that the graph $G'$ obtained from $G$ by applying a global rerouting operation  $(\g,\f)$
 contains neither twin nor whirl configurations. Thus, by Lemma~\ref{lem:fan}, $G'$ does not contain $3$-crossings and the statement follows.
\end{proof}

\section{Obtaining simplicity}\label{se:simplicity}

Lemmas~\ref{lem:nobigcrossings} and~\ref{lem:g} imply that, for $k \ge 2$, every \kplanarstg $G$ can be redrawn such that the resulting topological graph $G'$, $G' \simeq G$, contains no $(k+1)$-crossings and no two edges are rerouted around the same vertex. We now show that $G'$ is also simple if $k=2$. Then we handle the case $k \ge 3$, in which $G'$ may not be simple.

\begin{lemma}\label{lem:finally-simple}
For $k=2$ the topological graph $G'$ in Lemma~\ref{lem:g} is simple.
\end{lemma}
\begin{proof}
  Suppose, for a contradiction, that $G'$ is not simple, i.e., either two edges cross at least twice or two adjacent edges cross each other.

  Suppose first that there exists two edges $e_1$ and $e_2$ that cross at least twice in $G'$. A
  safe edge does not cross any edge more than once by
  Lemma~\ref{lem:nosafeedges}(ii). Any two nonrerouted edges cross at most once since $G$ is simple.
  A critical edge crosses any nonrerouted edge at
  most once by construction. It remains to consider the case that both $e_1$ and
  $e_2$ are critical, that is, $e_1=\g(X_1)$ and $e_2=\g(X_2)$ for some $3$-crossings
  $X_1,X_2\in \mathcal X$.  By Lemma~\ref{lem:cross}, $\g(X_1)$ is incident to $\f(X_2)$
  and $\g(X_2)$ is incident to $\f(X_1)$.  It follows that $(\f(X_1),\f(X_2))$ is a
  $2$-cycle in the conflict digraph $K$.
  By Lemma~\ref{lem:nested}, every $2$-cycle in $K$ is smooth, and by
  Lemma~\ref{lem:config}, the edges $\g(X_1)$ and $\g(X_2)$ do not cross in $G'$.
  This contradicts our assumption that $e_1=\g(X_1)$ and $e_2=\g(X_2)$ cross twice.
  We conclude that any two edges in $G'$ cross at most once.

  Suppose next that $e_1$ and $e_2$ are adjacent and cross at least once in
  $G'$. Two adjacent nonrerouted edges do not cross because $G$ is simple.  If a
  safe edge crosses an adjacent edge~$e'$, then~$e'$ is critical by
  Lemma~\ref{lem:nosafeedges}(iii).  Therefore, we may assume that $e_1$ is
  critical, that is, $e_1=\g(X_1)$ for some $X_1\in \mathcal X$. As two adjacent
  critical edges do not cross by Lemma~\ref{lem:noadjcross2}, only two cases
  remain.

  \paragraph{Case~1: $e_2$ is safe} Then by Lemma~\ref{lem:nosafeedges}(iii),
  $e_2$ is drawn in $\R{X_2}$, for some $X_2\in \mathcal X\setminus\{X_1\}$,
  $e_2$ is incident to $\f(X_1)$, and $e_1$ is incident to $\f(X_2)$.  Since
  $e_1$ and $e_2$ are adjacent, $e_2$ is also incident to an endpoint of
  $\g(X_1)$.  By the injectivity of $\f$, $\f(X_1)\neq \f(X_2)$, and so
  $e_2=\edge{\f(X_1)}{\f(X_2)}$. It follows that $(\f(X_1),\f(X_2))$ is a $2$-cycle
  in the conflict digraph $K$. By Lemma~\ref{lem:nested}, every $2$-cycle in~$K$
  is smooth. So $e_2$ is an edge between two consecutive vertices along
  $\partial\R{X_2}$, in contradiction to $X_2$ being a home for $e_2$.

  \paragraph{Case~2: $e_2$ is nonrerouted} By Lemma~\ref{lem:cross}, if
  $e_1=\g(X_1)$ crosses $e_2$, then $e_2$ is incident to $\f(X_1)$.  So
  $e_2=\edge{f(X_1)}{z}$, where $z$ is an endpoint of $\g(X_1)$ since $e_1$ and $e_2$ share an endpoint distinct from $f(X_1)$.
  If $e_2$ belongs to $\partial\R{X_1}$, then~$e_1$ crosses $e_2$ neither in $G$ nor in $G'$, contrary to our assumption that the two edges cross each other. Otherwise, the $3$-crossing $X_1$ is a
  home for $e_2$. But then $e_2$ would have been rerouted in a home rerouting
  operation, contrary to our assumption that $e_2$ is nonrerouted.

  Both cases lead to a contradiction and hence the statement follows.
\end{proof}

We now consider the case $k \ge 3$. We first characterize in the possible configurations that yield pairs of edges that intersect in two or more points in $G'$ (Lemma~\ref{lem:reroute-simple}), and then show how to further redraw some of these edges to eliminate multiple intersections without introducing  $(k+1)$-crossings (Lemma~\ref{lem:finally-nonsimple}).

\begin{lemma}\label{lem:reroute-simple}
Let $G$ be a \kplanarstg, where $k \ge 3$, and $G'$ the topological graph obtained from $G$ by executing a full rerouting operation.
\begin{itemize}\itemsep -2pt
\item If $e_1$ and $e_2$ are independent edges that cross more than once in $G'$, then both $e_1$ and $e_2$ are critical and they each are rerouted around an endpoint of the other
    (that is, $e_1=\g(X_1)$ and $e_2=\g(X_2)$ for some $(k+1)$-crossings $X_1,X_2\in \mathcal{X}$, $X_1\neq X_2$, and $\f(e_1)$ is an endpoint of $e_2$ and $\f(e_2)$ is an endpoint of $e_1$). Furthermore, $\f(X_1)$ and $\f(X_2)$ are adjacent in neither $\partial \R{X_1}$ nor $\partial \R{X_2}$.
%
\item If $e_1=\edge{u}{v}$ and $e_2=\edge{v}{w}$ are adjacent edges that cross in $G'$, then after possibly exchanging the roles of $e_1$ and $e_2$,
    we have that $e_1$ is critical rerouted around $w$, and $e_2$ is safe rerouted in its home which is the $(k+1)$-crossing of $e_1$. Furthermore, there exists a critical edge $e_3$ such that $e_1$ and $e_2$ have each been rerouted around an endpoint of the other.
\end{itemize} 
\end{lemma}
\begin{proof}
%
Assume first that there exist two edges, $e_1$ and $e_2$, crossing two or more times in $G'$. Since $G$ is simple, at least one of them, say $e_1$, has been rerouted. By Lemma~\ref{lem:nosafeedges}\eqref{nosafeedges-i}, neither~$e_1$ nor $e_2$ is safe, and thus $e_1$ is critical. Also, if $e_2$ is nonrerouted, $e_1$ crosses $e_2$ at most once by Lemma~\ref{lem:reroute-properties}\eqref{reroute-properties-iii}. Thus, we may assume that both $e_1$ and $e_2$ are critical. This implies that~$e_1$ and $e_2$ belong to different $(k+1)$-crossings of $G$; so, they do not cross in $G$ by Property~\ref{pr:k-crossings-disjoint}. Hence, by Lemma~\ref{lem:cross}(b), at least one of them has been rerouted around an endpoint of the other, say $e_1$ around an endpoint of $e_2$. This introduces a single crossing between $e_1$ and $e_2$, namely between the hook of $e_1$ and a tip of $e_2$. Thus, the other crossing must be between the hook of $e_2$ and a tip of $e_1$. By Lemma~\ref{lem:cross}(b) and by the injectivity of $\f$, no two edges are rerouted around the same vertex. We may assume that $e_1=\g(X_1)$ and $e_2=\g(X_2)$ for some $(k+1)$-crossings $X_1,X_2\in \mathcal{X}$, $X_1\neq X_2$; and  $\f(e_1)$ is an endpoint of $e_2$ and $\f(e_2)$ is an endpoint of $e_1$. If vertices $\f(e_1)$ and $\f(e_2)$ are adjacent in $\partial \R{X_1}$ (resp., $\partial \R{X_2}$), then the tip of $e_1$ (resp., $e_2$) follows the edge $\edge{\f(X_1)}{\f(X_2)}$, consequently avoids $e_2$ (resp, $e_1$), and so $e_1$ and $e_2$ would cross at most once (see \figurename~\ref{fig:fan} and \figurename~\ref{fig:twin:a} for examples).

Assume now that there exist two adjacent edges, $e_1=\edge{u}{v}$ and $e_2=\edge{v}{w}$, that cross in $G'$. 
Since $G$ is simple and by Lemma~\ref{lem:nosafeedges}\eqref{nosafeedges-iii}, at least one of $e_1$ and $e_2$ is critical and rerouted around an endpoint of the other. In particular, we can assume that $e_1$ is critical and rerouted around the vertex $u$. But then the $(k+1)$-crossing of $e_1$ is a home for $e_2$. In particular, $e_2$ is safe. By Lemma~\ref{lem:nosafeedges}\eqref{nosafeedges-iii}, there exists another $(k+1)$-crossing that is a home for $e_2$, which implies the existence of a critical edge $e_3$ incident to $u$ that has been rerouted around $v$, as claimed.
\end{proof}

\begin{figure}
	\centering
	\begin{minipage}[b]{.4\textwidth}
		\centering
		\includegraphics[page=2,scale=1]{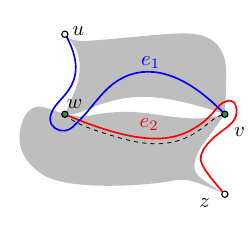}
		\subcaption{~}\label{fig:non-simple-endpoints}
	\end{minipage}
	\hfil
	\begin{minipage}[b]{.4\textwidth}
	\centering
		\includegraphics[page=6,scale=1]{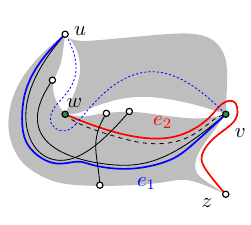}
		\subcaption{~}\label{fig:non-simple-endpoints-u-in-xd}
	\end{minipage}\\
	\begin{minipage}[b]{.4\textwidth}
		\centering
		\includegraphics[page=4,scale=1]{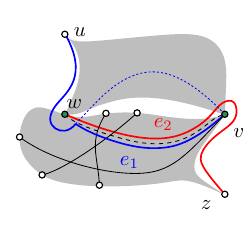}
		\subcaption{~}\label{fig:non-simple-rerouting}
	\end{minipage}
	\hfil
	\begin{minipage}[b]{.4\textwidth}
		\centering
		\includegraphics[page=5,scale=1]{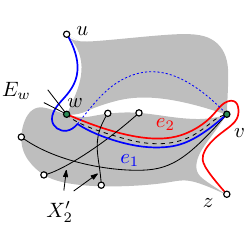}
		\subcaption{~}\label{fig:non-simple-rerouting-crossings}
	\end{minipage}
	\caption{%
          (a)~A double crossing between two edges $e_1$ and $e_2$ due to
          rerouting; the dashed edge $\edge{v}{w}$ may be present or not.  The
          configuration is resolved by redrawing the edge $e_1$ as in (b)~if
          $u\in\V{X_2}$ or (c)~if $u\notin\V{X_2}$.  (d)~Edges crossing $e_1$
          after the transformation.}
\end{figure}

\begin{lemma}\label{lem:finally-nonsimple}
For every \kplanarstg $G$, where $k \ge 3$, there exists a \kplusquasistg  $G^*$, $G^* \simeq G$.
\end{lemma}
\begin{proof}
Let $G'$ be the  \kplusquasi topological graph obtained from $G$ by executing a full rerouting operation (cf.~Lemma~\ref{lem:nobigcrossings}). We may assume that $G'$ is not simple, as otherwise the statement would follow with $G^*=G'$. By Lemma~\ref{lem:reroute-simple}, there exists a pair of critical edges that are each rerouted around an endpoint of the other; see \figurename~\ref{fig:non-simple-endpoints}. Let $\mathcal{P}$ be the set of such (unordered) pairs.
Note that the pairs in $\mathcal{P}$ are pairwise disjoint since $\f$ is injective.

\paragraph{Rerouting operations}
For every pair in $\mathcal{P}$, we reroute one of the two edges as follows. Let $\{e_1,e_2\}\in \mathcal{P}$.
We introduce some notation. Assume that $e_1\in X_1$ and $e_2\in X_2$ for some $X_1,X_2\in \mathcal{X}$.
Further assume that $e_1=\edge{u}{v}$ is rerouted around vertex $w$, and $e_2=\edge{w}{z}$ is rerouted around $v$.
Note that $v$ and $w$ are adjacent in neither $\partial \R{X_1}$ nor $\partial \R{X_2}$
by Lemma~\ref{lem:reroute-simple}. 
Hence $u$ and $w$ are consecutive along $\partial \R{X_1}$, and $v$ and $z$ are consecutive along $\partial \R{X_2}$.

If $\edge{v}{w}\notin E$, then we can arbitrarily choose $e_1$ or $e_2$ to be redrawn. Otherwise, edge $\edge{v}{w}$ has a home in both $X_1$ and $X_2$. If $\edge{v}{w}$ has been redrawn in $X_1$ due to a home rerouting operation, then we redraw $e_1$, else we redraw $e_2$. In the following we assume without loss of generality that we redraw $e_1$.

We distinguish between two cases, based on whether $u \in \V{X_2}$. Assume first that $u \in \V{X_2}$. In this case, both endpoints of $e_1=\edge{u}{v}$ are in $\V{X_d}$ for $d=1,2$, but $e_1\notin X_2$ and $X_2$ is not a home for $e_1$ (otherwise $e_1$ would be safe). Similarly to the home rerouting operation, we redraw $e_1$ in the region $\R{X_2}\setminus \D{X_2}$, such that it crosses the edges of $X_2$ at most once, and such that it crosses neither $e_2$ nor the edge of $X_2$ incident to $u$; see \figurename~\ref{fig:non-simple-endpoints-u-in-xd}. Note that the new drawing of $e_1$ does not cross the two possible safe edges in $\R{X_2}$.

Assume next that $u \notin \V{X_2}$. We redraw the portion of $e_1$ between the two crossings with $e_2$ by following $e_2$, crossing neither $e_2$ nor the possible safe edge $\edge{v}{w}$ (if it exists). More precisely, we redraw the tip of $e_1$ crossed by the hook of $e_2$ by following the tip of $e_2$ crossed by the hook of $e_1$ without crossing it and without crossing $\edge{v}{w}$; see \figurename~\ref{fig:non-simple-rerouting}.

Denote by $G^*$ the topological graph obtained by applying the operation for every pair in $\mathcal{P}$.

\paragraph{Proving that $G^*$ does not contain $(k+1)$-crossings}
Since $G'$ is \kplusquasi by Lemma~\ref{lem:nobigcrossings}, every $(k+1)$-crossings
of $G^*$ involves an edge that has been redrawn. 
Let $e_1$ be an edge that has been redrawn from a pair $\{e_1,e_2\}\in \mathcal{P}$,
let $X_1$ and $X_2$ be the $(k+1)$-crossings of $G$ containing $e_1$ and $e_2$, respectively.
The edges crossing $e_1$ in $G^*$ are (see \figurename~\ref{fig:non-simple-rerouting-crossings}):
\begin{enumerate}[(i)]
\item a set $X_2'\subset X_2$ of edges crossing the tip of $e_2$ that is used to
  redraw $e_1$ and
\item a fan $E_w$ of edges incident to the vertex $w$ around which $e_1$ has been
  rerouted (and thus these edges cross the hook of $e_1$).
\end{enumerate}
Since $u$ and $w$ are consecutive around $\R{X_1}$, it follows that $e_1$ does not
cross any edge in $X_1$. Further, note that $X_2'$ contains all the edges that cross
$e_1$ in $G^*$ and not in $G'$. This immediately implies that any $(k+1)$-crossing
of $G^*$ that contains $e_1$ also contains some edge of $X_2'$, otherwise such a
$(k+1)$-crossing would also exist in $G'$. The edges in $X_2'$ do not cross
edges in $E_w$, since $X_2$ does not contain any edge incident to $w$, other
than $e_2$. Finally, there are at most $k-1$ edges in~$X_2'$, since $X_2$ contains
$k+1$ edges and at least two of them do not cross $e_1$, namely $e_2$ and the edge
incident to $v$. Thus, the edges in~$X_2'$ are not involved in any
$(k+1)$-crossing with $e_1$. This implies that $G^*$ is $(k+1)$-quasiplanar.

\paragraph{Proving that $G^*$ is simple}
By Lemma~\ref{lem:reroute-simple} and by the choice of the edge $e_1\in \{e_1,e_2\}$, $\{e_1,e_2\}\in \mathcal{P}$,
if $G^*$ contains a pair of independent edges that cross more then once or a pair of adjacent edges cross,
then this pair must include an edge that has been rerouted. 

Let $e_1$ be an edge that has been redrawn from a pair $\{e_1,e_2\}\in \mathcal{P}$, with the notation 
used for the description of the rerouting above. 
First observe that $e_1$ does not cross any edge incident to $v$ in $G^*$: it crosses neither $\edge{v}{w}$ nor the edge of $X_2$ incident to $v$, by construction. Also, $e_1$ does not cross any edge incident to $u$ in $G^*$. In fact, $e_1$ does not cross $\edge{u}{w}$, since $u$ and $w$ are consecutive along $\partial \R{X_1}$; hence, if $e$ crosses an edge incident to $u$, then this edge belongs to~$X_2'$. However, this implies that $u \in \V{X_2}$, and thus $e_1$ has been redrawn without crossing the edge of $X_2$ incident to $u$.

It remains to prove that $e_1$ does not cross any edge more than once. First observe that
the tip of $e_1$ incident to $u$ has not been redrawn, since it has been rerouted in
$G'$ by following edge $\edge{u}{w}$, which is uncrossed in $G$ since $u$ and $w$ are
consecutive along $\partial \R{X_1}$. Also, $e_1$ does not cross any edge in $E_w$ twice,
since it crosses all these edges in $G'$, and the only edge crossed twice by $e_1$
in $G'$ is $e_2$, by Lemma~\ref{lem:reroute-simple}. Hence, if $e_1$ crosses an edge
$e_3$ twice, then $e_3$ belongs to~$X_2'$. Since $e_1$ does not cross any edge of
$X_2'$ in $G'$, by Property~\ref{pr:k-crossings-disjoint}, we have that both
crossings between $e_1$ and $e_3$ have been introduced by the redrawing of
$e_1$. However, if $u \in \V{X_2}$, then the redrawing of $e_1$, which is analogous
to the home rerouting operation, does not introduce any double crossing on $e_1$,
by construction. On the other hand, if $u \notin \V{X_2}$, then the part of $e_1$
that has been redrawn has the same crossings as $e_2$, which does not cross any
edge in $X_2$ twice. Thus, $e_1$ does not cross any edge twice.
We conclude that $G^*$ is a simple topological graph.
\end{proof}

The next theorem summarizes the main result of the paper.

\begin{theorem}\label{th:result}
Every \kplanarstg is a \kplusquasistg, for every $k \geq 2$.
\end{theorem}
\begin{proof}
Let $G$ be a \kplanarstg, with $k \geq 2$. We show that there exists a \kplusquasistg $G^*$ such that $G^* \simeq G$.	
First recall that, by Lemma~\ref{lem:removing-tangled}, we can assume that $G$ does not contain any tangled $(k+1)$-crossing. By Lemma~\ref{lem:g}, for $k=2$, and by Lemma~\ref{lem:nobigcrossings}, for $k \geq 3$, there exist functions $\f$ and $\g$ such that the topological graph $G' \simeq G$ obtained by applying a global rerouting  $(\g,\f)$ to $G$ is \kplusquasi.

When $k=2$, the topological graph $G'$ is also simple, as proved in Lemma~\ref{lem:finally-simple}, and the statement follows with $G^* = G'$. For the case $k \geq 3$, instead, $G'$ need not be simple. However, in this case, we can apply an additional redrawing operation which results in a \kplusquasistg $G^*$, $G^* \simeq G'$, by Lemma~\ref{lem:finally-nonsimple}.
\end{proof}

\section{Conclusions and Open Problems}\label{sec:conclusions}

\begin{figure}
	\centering
	\includegraphics[width=0.8\textwidth]{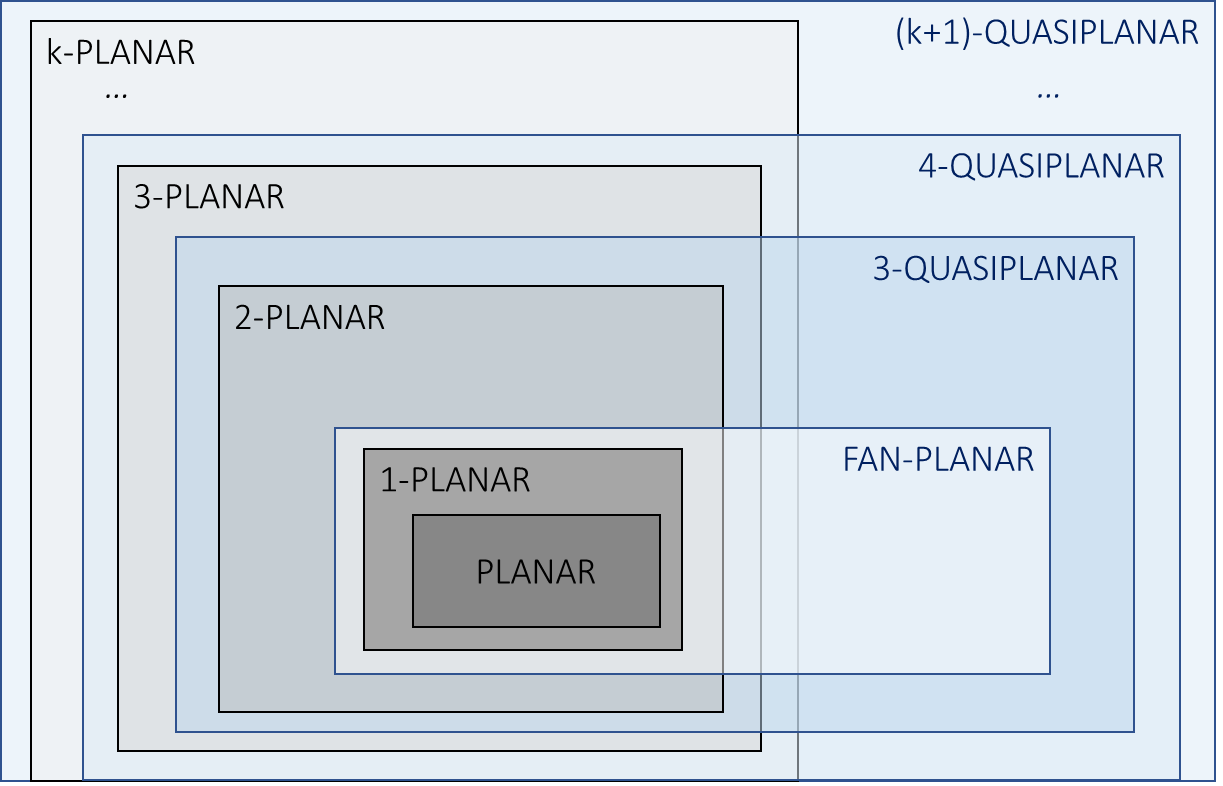}
	\caption{A Venn diagram showing that every \kplanarstg is a \kplusquasistg, for every $k \geq 2$.\label{fi:venn}}
\end{figure}

We have proved that, for any $k \geq 2$, the family of \kplanarstg{s} is included in the family of \kplusquasistg{s}, see also \figurename~\ref{fi:venn} for a diagram illustrating this relationship. This result represents the first nontrivial relationship between the $k$-planar and the $k$-quasiplanar graph hierarchies, and contributes to the literature that studies the connection between different families of beyond planar graphs (see, e.g.~\cite{DBLP:journals/jgaa/BinucciCDGKKMT17,DBLP:journals/tcs/BinucciGDMPST15,DBLP:journals/tcs/BrandenburgDEKL16,DBLP:journals/dam/EadesL13}). Several interesting problems remain open. Among them:

\begin{itemize}
\item We do not know whether our main result extends to nonsimple \planar{k} graphs for $k\geq 4$: Is every \planar{k} graph \quasi{(k+1)}? For $k\leq 3$, every \planar{k} graph is simple \planar{k}~\cite[Lemma~1.1]{DBLP:journals/dcg/PachRTT06}.
    Hence Theorem~\ref{th:result} readily implies that every \planar{k} graph is simple \quasi{(k+1)} for $k=2,3$.

\item For $k \geq 3$, one can also ask whether every $k$-planar graph is $k$-quasiplanar.
 For $k=2$ the answer is trivially negative, as $2$-quasiplanar graphs are precisely the planar graphs. On the other hand, optimal 3-planar graphs are known to be ($3$-)quasiplanar~\cite{bkr-socg17}. We recall that an $n$-vertex $3$-planar graph is optimal if it has $5.5n-11$ edges~\cite{DBLP:conf/jcdcg/PachRT02}, and that so far only multi-graphs are known to be in this family, while it is still unknown whether $3$-planar graphs with no multi-edges can have $5.5n-11$ edges (see also~\cite{bkr-socg17}). For sufficiently large values of $k$, one can even investigate whether every $k$-planar simple topological (sparse) graph $G$ is $f(k)$-quasiplanar, for some function $f(k)=o(k)$.

\item One can study \emph{non-inclusion} relationships between the $k$-planar and the $k$-quasiplanar graph hierarchies, other than those that are easily derivable from the known edge density results. For example, for any given $k > 3$, can we establish an integer function $h(k)$ such that some \planar{h(k)} graph is not \quasi{k}?

\item A long-standing open problem is to establish the computational complexity of recognizing $k$-quasiplanar graphs. Is there a polynomial-time algorithm that decides whether a given graph is quasiplanar (or $k$-quasiplanar for a given constant $k$)?

\end{itemize}

\subsection*{Acknowledgements}
The research in this paper started at the Dagstuhl Seminar 16452 ``Beyond-Planar
Graphs: Algorithmics and Combinatorics'' and at the Fifth Annual Workshop on Geometry and Graphs,
March 6--10, 2017, at the Bellairs Research Institute of McGill University.
We thank all participants, and in particular Pavel Valtr and Raimund Seidel, for useful discussions on the topic.

Research supported in part by: DFG grant Ka812/17-1, MIUR project ``MODE'' under PRIN 20157EFM5C, ``AHeAD'' under PRIN 20174LF3T8, H2020-MSCA-RISE project 734922 – ``CONNECT'', MIUR-DAAD Joint Mobility Program: N$^\circ$ 34120 and N$^\circ$ 57397196, project ``Algoritmi e sistemi di analisi visuale di reti complesse e di grandi dimensioni - Ricerca di Base 2018, Dipartimento di Ingegneria dell'Universit\`a degli Studi di Perugia'', by the Swiss National Science Foundation within the collaborative DACH project \emph{Arrangements and Drawings} as SNSF Project 200021E-171681, and by the NSF awards CFF-1422311 and CFF-1423615.

Finally, we would like to acknowledge the anonymous referees of this paper for their useful comments and suggestions, which helped us to improve the quality of this paper.

\bibliographystyle{abbrv}
\bibliography{bibliography}

\begin{thebibliography}{10}

\bibitem{DBLP:journals/dcg/Ackerman09}
E.~Ackerman.
\newblock On the maximum number of edges in topological graphs with no four
  pairwise crossing edges.
\newblock {\em Discrete Comput. Geom.}, 41(3):365--375, 2009.

\bibitem{DBLP:journals/corr/Ackerman15}
E.~Ackerman.
\newblock On topological graphs with at most four crossings per edge.
\newblock {\em CoRR}, abs/1509.01932, 2015.

\bibitem{DBLP:journals/comgeo/AckermanFPS14}
E.~Ackerman, J.~Fox, J.~Pach, and A.~Suk.
\newblock On grids in topological graphs.
\newblock {\em Comput. Geom.}, 47(7):710--723, 2014.

\bibitem{DBLP:journals/jct/AckermanT07}
E.~Ackerman and G.~Tardos.
\newblock On the maximum number of edges in quasi-planar graphs.
\newblock {\em J. Comb. Theory, Ser. {A}}, 114(3):563--571, 2007.

\bibitem{DBLP:journals/combinatorica/AgarwalAPPS97}
P.~K. Agarwal, B.~Aronov, J.~Pach, R.~Pollack, and M.~Sharir.
\newblock Quasi-planar graphs have a linear number of edges.
\newblock {\em Combinatorica}, 17(1):1--9, 1997.

\bibitem{DBLP:conf/wg/AngeliniBBLBDLM17}
P.~Angelini, M.~A. Bekos, F.~J. Brandenburg, G.~{Da Lozzo}, G.~{Di Battista},
  W.~Didimo, G.~Liotta, F.~Montecchiani, and I.~Rutter.
\newblock On the relationship between k-planar and k-quasi-planar graphs.
\newblock In {\em Graph-Theoretic Concepts in Computer Science}, volume 10520
  of {\em LNCS}, pages 59--74. Springer, 2017.

\bibitem{BAE2018}
S.~W. Bae, J.-F. Baffier, J.~Chun, P.~Eades, K.~Eickmeyer, L.~Grilli, S.-H.
  Hong, M.~Korman, F.~Montecchiani, I.~Rutter, and C.~D. T\'oth.
\newblock Gap-planar graphs.
\newblock {\em Theoretical Computer Science}, 2018.

\bibitem{DBLP:conf/gd/Bekos0R16}
M.~A. Bekos, M.~Kaufmann, and C.~N. Raftopoulou.
\newblock On the density of non-simple 3-planar graphs.
\newblock In {\em Graph Drawing and Network Visualization}, volume 9801 of {\em
  LNCS}, pages 344--356. Springer, 2016.

\bibitem{bkr-socg17}
M.~A. Bekos, M.~Kaufmann, and C.~N. Raftopoulou.
\newblock On optimal 2- and 3-planar graphs.
\newblock In {\em Symposium on Computational Geometry}, volume~77 of {\em
  LIPIcs}, pages 16:1--16:16. Schloss Dagstuhl - Leibniz-Zentrum fuer
  Informatik, 2017.

\bibitem{DBLP:journals/jgaa/BinucciCDGKKMT17}
C.~Binucci, M.~Chimani, W.~Didimo, M.~Gronemann, K.~Klein,
  J.~Kratochv{\'{\i}}l, F.~Montecchiani, and I.~G. Tollis.
\newblock Algorithms and characterizations for 2-layer fan-planarity: From
  caterpillar to stegosaurus.
\newblock {\em J. Graph Algorithms Appl.}, 21(1):81--102, 2017.

\bibitem{DBLP:journals/tcs/BinucciGDMPST15}
C.~Binucci, E.~{Di Giacomo}, W.~Didimo, F.~Montecchiani, M.~Patrignani,
  A.~Symvonis, and I.~G. Tollis.
\newblock Fan-planarity: Properties and complexity.
\newblock {\em Theor. Comput. Sci.}, 589:76--86, 2015.

\bibitem{DBLP:journals/tcs/BrandenburgDEKL16}
F.~J. Brandenburg, W.~Didimo, W.~S. Evans, P.~Kindermann, G.~Liotta, and
  F.~Montecchiani.
\newblock Recognizing and drawing {IC}-planar graphs.
\newblock {\em Theor. Comput. Sci.}, 636:1--16, 2016.

\bibitem{DBLP:journals/jct/CapoyleasP92}
V.~Capoyleas and J.~Pach.
\newblock A {T}ur{\'{a}}n-type theorem on chords of a convex polygon.
\newblock {\em J. Comb. Theory, Ser. {B}}, 56(1):9--15, 1992.

\bibitem{DBLP:journals/algorithmica/CheongHKK15}
O.~Cheong, S.~Har{-}Peled, H.~Kim, and H.~Kim.
\newblock On the number of edges of fan-crossing free graphs.
\newblock {\em Algorithmica}, 73(4):673--695, 2015.

\bibitem{DBLP:journals/csur/DidimoLM19}
W.~Didimo, G.~Liotta, and F.~Montecchiani.
\newblock A survey on graph drawing beyond planarity.
\newblock {\em {ACM} Comput. Surv.}, 52(1):4:1--4:37, 2019.

\bibitem{DBLP:journals/dam/EadesL13}
P.~Eades and G.~Liotta.
\newblock Right angle crossing graphs and 1-planarity.
\newblock {\em Discrete Appl. Math.}, 161(7-8):961--969, 2013.

\bibitem{DBLP:conf/compgeom/FoxP08}
J.~Fox and J.~Pach.
\newblock Coloring {$K_k$}-free intersection graphs of geometric objects in the
  plane.
\newblock In {\em Symposium on Computational Geometry}, pages 346--354. {ACM},
  2008.

\bibitem{DBLP:journals/siamdm/FoxPS13}
J.~Fox, J.~Pach, and A.~Suk.
\newblock The number of edges in k-quasi-planar graphs.
\newblock {\em {SIAM} J. Discrete Math.}, 27(1):550--561, 2013.

\bibitem{DBLP:conf/mfcs/0001T17}
M.~Hoffmann and C.~D. T{\'{o}}th.
\newblock Two-planar graphs are quasiplanar.
\newblock In {\em {MFCS}}, volume~83 of {\em LIPIcs}, pages 47:1--47:14.
  Schloss Dagstuhl - Leibniz-Zentrum fuer Informatik, 2017.

\bibitem{DBLP:journals/corr/KaufmannU14}
M.~Kaufmann and T.~Ueckerdt.
\newblock The density of fan-planar graphs.
\newblock {\em CoRR}, abs/1403.6184, 2014.

\bibitem{DBLP:journals/gc/PachPST05}
J.~Pach, R.~Pinchasi, M.~Sharir, and G.~T{\'{o}}th.
\newblock Topological graphs with no large grids.
\newblock {\em Graphs and Combinatorics}, 21(3):355--364, 2005.

\bibitem{DBLP:journals/dcg/PachRTT06}
J.~Pach, R.~Radoi{\v{c}}i{\'{c}}, G.~Tardos, and G.~T{\'{o}}th.
\newblock Improving the crossing lemma by finding more crossings in sparse
  graphs.
\newblock {\em Discrete Comput. Geom.}, 36(4):527--552, 2006.

\bibitem{DBLP:conf/jcdcg/PachRT02}
J.~Pach, R.~Radoi{\v{c}}i{\'{c}}, and G.~T{\'o}th.
\newblock Relaxing planarity for topological graphs.
\newblock In {\em Japanese Conf. Discrete Comput. Geom.}, volume 2866 of {\em
  LNCS}, pages 221--232. Springer, 2003.

\bibitem{DBLP:journals/algorithmica/PachSS96}
J.~Pach, F.~Shahrokhi, and M.~Szegedy.
\newblock Applications of the crossing number.
\newblock {\em Algorithmica}, 16(1):111--117, 1996.

\bibitem{DBLP:journals/combinatorica/PachT97}
J.~Pach and G.~T{\'o}th.
\newblock Graphs drawn with few crossings per edge.
\newblock {\em Combinatorica}, 17(3):427--439, 1997.

\bibitem{DBLP:journals/comgeo/SukW15}
A.~Suk and B.~Walczak.
\newblock New bounds on the maximum number of edges in $k$-quasi-planar graphs.
\newblock {\em Comput. Geom.}, 50:24--33, 2015.

\bibitem{DBLP:journals/dcg/Valtr98}
P.~Valtr.
\newblock On geometric graphs with no $k$ pairwise parallel edges.
\newblock {\em Discrete Comput. Geom.}, 19(3):461--469, 1998.

\end{thebibliography}

\end{document}